\documentclass[prx,reprint,showpacs]{revtex4-1}

\usepackage{amsfonts,amssymb,amsmath}
\usepackage[]{graphics,graphicx,epsfig,wrapfig}
\usepackage{sidecap}
\usepackage{amsthm}
\usepackage{float}
\usepackage{enumitem}
\usepackage{color}
\usepackage{caption,subcaption}
\usepackage[noxcolor]{pstricks} 
\usepackage{pst-grad} 
\usepackage{pst-plot} 
\usepackage[color]{changebar}
\usepackage{booktabs}
\usepackage{multirow}
\usepackage{siunitx}
\usepackage{verbatim}
\usepackage{mathpazo}
\usepackage{csquotes}

\definecolor{myurlcolor}{rgb}{0,0,0.7}
\definecolor{myrefcolor}{rgb}{0.8,0,0}
\usepackage[unicode=true,pdfusetitle, bookmarks=true,bookmarksnumbered=false,bookmarksopen=false, breaklinks=false,pdfborder={0 0 0},backref=false,colorlinks=true, linkcolor=myrefcolor,citecolor=myurlcolor,urlcolor=myurlcolor]{hyperref}

\newtheorem{definition}{Definition}
\newtheorem{fact}{Fact}

\newcommand{\be}{\begin{equation}}
\newcommand{\ee}{\end{equation}}
\newcommand{\bei}{\begin{itemize}}
\newcommand{\eei}{\end{itemize}}
\newcommand{\beq}{\begin{eqnarray}}
\newcommand{\eeq}{\end{eqnarray}}
\newcommand{\ben}{\begin{eqnarray}}
\newcommand{\een}{\end{eqnarray}}
\newcommand{\bea}{\begin{array}}
\newcommand{\eea}{\end{array}}
\newcommand{\<}{\langle}
\renewcommand{\>}{\rangle}

\def\squareforqed{\hbox{\rlap{$\sqcap$}$\sqcup$}}
\def\qed{\ifmmode\squareforqed\else{\unskip\nobreak\hfil
\penalty50\hskip1em\null\nobreak\hfil\squareforqed
\parfillskip=0pt\finalhyphendemerits=0\endgraf}\fi}
\def\endenv{\ifmmode\;\else{\unskip\nobreak\hfil
\penalty50\hskip1em\null\nobreak\hfil\;
\parfillskip=0pt\finalhyphendemerits=0\endgraf}\fi}

\def\ot{\otimes}
\def\precom{independence}
\def\Precom{Independence}
\def\prec{Ind}
\def\com{complementarity}
\def\Com{Complementarity}

\newcommand{\eqdef}{\mathrel{:=}} 

\newcommand{\setgen}{S} 
\newcommand{\setth}{S_{X,Y}}  
\newcommand{\sxy}{\setth}

\renewcommand{\P}{\mathcal{P}} 
\newcommand{\M}{\mathcal{M}} 
\newcommand{\q}{\mathbf{q}} 

\renewcommand{\U}{\mathrm{U}} 
\renewcommand{\C}{\mathrm{C}} 
\newcommand{\E}{\mathrm{E}} 

\newcommand{\mg}{\color{black}}  

\newcommand{\blk}{\color{black}}

 \theoremstyle{plain}
 \newtheorem{lem}{Lemma}
 \theoremstyle{plain}
 
 \theoremstyle{plain}
 
 \theoremstyle{plain}
 
 \theoremstyle{plain}
 
  \theoremstyle{plain}
 
 \theoremstyle{remark}
 \newtheorem*{rem*}{Remark}
 \theoremstyle{plain}

  \newcommand{\ketbra}[2]{| #1 \rangle \langle #2 |}
\newcommand{\ket}[1]{| #1 \rangle}

\newcommand{\braket}[2]{\langle #1 | #2 \rangle}

\def\ur{UR}
\def\pur{PUR}

\begin{document}

\title{Operational foundations for complementarity and uncertainty relations}

\author{Debashis Saha}
\email{saha@cft.edu.pl}
\affiliation{ 
Institute of Theoretical Physics and Astrophysics, National Quantum Information Centre, Faculty of Mathematics, Physics and
Informatics, University of Gdansk, Wita Stwosza 57, 80-308 Gda\'nsk, Poland}
\affiliation{Center for Theoretical Physics, Polish Academy of Sciences, Al. Lotnik\'{o}w 32/46, 02-668 Warsaw, Poland}

\author{Micha\l\ Oszmaniec}
\email{michal.oszmaniec@gmail.com}
\affiliation{ 
Institute of Theoretical Physics and Astrophysics, National Quantum Information Centre, Faculty of Mathematics, Physics and
Informatics, University of Gdansk, Wita Stwosza 57 , 80-308 Gda\'nsk, Poland}

\author{Lukasz Czekaj}
\email{jasiek.gda@gmail.com}
\affiliation{Faculty of Applied Physics and Mathematics, National Quantum Information Centre, Gda\'nsk University of Technology, 
80-233 Gda\'nsk, Poland}

\author{Micha\l\ Horodecki}
\email{fizmh@ug.edu.pl }
\affiliation{ 
Institute of Theoretical Physics and Astrophysics, National Quantum Information Centre, Faculty of Mathematics, Physics and
Informatics, University of Gdansk, Wita Stwosza 57, 80-308 Gda\'nsk, Poland}

\author{Ryszard Horodecki}
\email{fizrh@ug.edu.pl }
\affiliation{ 
Institute of Theoretical Physics and Astrophysics, National Quantum Information Centre, Faculty of Mathematics, Physics and
Informatics, University of Gdansk, Wita Stwosza 57, 80-308 Gda\'nsk, Poland}


\begin{abstract}
The so-called preparation  uncertainty that occurs in quantum world can be understood well in purely operational terms, 
and its existence in any given theory, perhaps different than quantum mechanics, can be  verified by
examining only measurement statistics. Namely, one says that uncertainty occurs in some theory, when for some pair of observables,  there is no preparation, which would exhibit deterministic statistics for both of them.  
However the right hand side of uncertainty relation, is not operational anymore, if we do not insist,
that it is just minimum of the left hand side for a given theory. E.g. in quantum mechanics, it is some function of two observables, that must be computed within the quantum formalism.  Also, while joint non-measurability of observables is an operational notion, the complementarity in Bohr sense (i.e. in terms of information 
needed to describe the system) has not yet been expressed in purely operational terms. 

In this paper  we propose a solution to these two problems, by introducing  an operational definition for complementarity, 
and further postulating, that complementary observables have to exhibit uncertainty. In other words, we propose to put 
the (operational) complementarity as the right hand side of uncertainty relation. We thus view uncertainty as a necessary price 
for complementarity in physical theories. 

In more detail, we first identify two different notions of uncertainty and complementarity for which the above principle holds in the  quantum mechanical realm. We also introduce postulates for the general measures of uncertainty and complementarity. In order to define quantifiers of complementarity we first turn to the simpler notion of independence that is defined solely in terms of the statistics of two observables. Importantly, for clean \mg and extremal observables  - i.e. ones that cannot be simulated irreducibly by other observables  - \blk
any  measure of independence reduces to the proper complementary measure.   

Finally, as application of our general framework we define a number of complementarity indicators based on  (i) performance of random access codes, (ii) geometrical properties of the body of observed statistics, and  (iii) variation of information. 
We analyze the properties of these indicators and show that they can be used to state uncertainty relations. Moreover, we apply the uncertainty relation expressed by complementarity of type (ii) to show, how, under some natural symmetries, 
it leads to the Tsirelson bound for CHSH inequality. 
Lastly, we show that for a single system a variant of Information Causality called Information Content Principle, under the above symmetries,
can be interpreted as  uncertainty relation in the above sense. 

\end{abstract}

\maketitle

\section{Introduction}

Uncertainty and complementarity are landmark features of quantum mechanics and have been investigated since its inception almost a century ago. The concept complementarity captures the fact that in quantum mechanics two quantum observables cannot be measured simultaneously and hence supply "independent" informations about a physical systems \cite{Conmpl1928}. The uncertainty principle, proposed for the first time by Heinsenberg, on the other hand, limits the precision of outcome statistics of two complementary observables, like position and momentum \cite{Heisenberg1927}. Uncertainty relations are quantitative emanations of the  uncertainty principle and play predominant role in the conceptual \cite{BuschHeinsenberg2007} and mathematical foundations of quantum theory \cite{Birula1975,Bialynicki2011,Partovi-major,FriedlandGour2013,PochalaMajorisation2013}. Importantly, with the advent of quantum information, uncertainty relations found also practical applications in fields such as entanglement detection \cite{Giovannetti2004,Ghune2004} quantum steering \cite{Walborn2009}, as well as randomness generation and  quantum cryptography \cite{RevModPhys2017ent}.  

Despite the great success of the research effort concerning uncertainty relations,  this line of research is inherently restricted to quantum formalism and so the notions of complementarity, uncertainty, uncertainty-relations have not been  much   explored outside quantum theory.
 Uncertainty itself is defined pretty operationally, and it was explored in more general setup than quantum 
(see e.g. \cite{SteegWehner,OW,RaviGMP-ur-nonlocality,WehnerHanggi-up-thermo}).
The uncertainty relations were also considered in those papers. However,  the right hand sides of these relations were not expressed in operational terms. Also, while the issue of joint non-measurability (incompatibility)
 was explored outside of quantum mechanical formalism \cite{Plavala-incom,JencovaPlavala-incom,JanottaH-review,Busch13,Filippov17}, the complementarity 
of observables, understood in Bohr's sense seems not investigated so far in operational terms (apart from the approach, 
where complementarity is simply understood just the minimum of the left hand side of the uncertainty relation cf. \cite{BuschHeinsenberg2007}) 

This article aims to change this state of affairs. By defining  complementarity in purely operational fashion i.e. solely in terms of the statistics of measurements a given theory, we are able to obtain operational form of uncertainty relation - where both side of inequalities 
are some functions of just statistics of observables -  without referring to internal formalism of the theory.


The complementarity  should be associated to the independent information that can be obtained from two different observables which cannot be measured jointly. This notion of complementarity is motivated by Bohr's own views concerning this concept. In one of the letters to Einstein \cite{Bohr1949}   Bohr diatribes complementarity in the following words \footnote{See also \cite{Plotonisky2014} for the comprehensive account on on the evolution of Bohr's views on the notions of uncertainty and complementarity.}. 
\begin{displayquote}
Evidence obtained under different experimental
conditions cannot be comprehended within a single picture,
but must be regarded as complementary in the sense that only
the totality of the phenomena exhaust the possible information
about the objects. 
\end{displayquote} 
 
Furthermore, we postulate, inspired by quantum mechanics, that in reasonable physical theories uncertainty should be present for all complementary observables (i.e. we identify right hand side of uncertainty relation with complementarity).  In other words, uncertainty  should be regarded as a price that we pay for complementarity of two observables. This fundamental trade-off we refer to
as \emph{uncertainty principle}. On the other hand, for maximally informative measurements uncertainty should also imply complementarity. Finally, these two fundamental trade-offs are captured by \textit{uncertainty relations} and \textit{reverse uncertainty relations}  in a theory.

Let us  now outline the somehow unusual structure of this work. First, in Section \ref{sec:framework}, we present the general operational framework in which we cast concepts of uncertainty and complementarity.   Then, in Section \ref{sec:CONCEPT} we present the connections between various kinds of complementarity and uncertainty in quantum theory.   Importantly, we observe that in quantum mechanics there are three different notions of uncertainty. We observe that  for two of them there exist different but operationally well-motivated notions of complementarity that can be used to formulate \emph{uncertainty principles}.  We propose that in all reasonable physical theories the analogues of the aforementioned uncertainty principles should hold. In Section \ref{sec:basic-indep} we argue that often complementarity of two clean \mg and extremal observables  (i.e. ones that cannot be simulated irreducibly by other observables) \blk can be defined solely in terms of their output statistics. This is a great simplification as it allows to (in some cases) discuss complemantarity without any direct reference to the formalism or the structure of a particular theory. In Section \ref{sec:twoOUTCOMES} we present the intuitive exposition of our ideas in the case of dichotomic observables. This simple setting allows for a nice geometrical interpretation of our ideas concerning uncertainty, complementarity and uncertainty principle. After the first part of the paper, that has a rather  introductory and conceptual flavor, in Section \ref{sec:outline2} we give an overview and motivation for  technical results given latter the manuscript. Sections \ref{sec:post-uncert} and  \ref{sec:post-compl} present our postulates for measures of uncertainty as well as   complementarity and independence  respectively. In the latter Section \ref{sec:measures-uncert} we propose a number of concrete measures of  uncertainty and complementarity, that are motivated either by the operational or geometrical considerations. Finally, in Section \ref{sec:pur}  we use some of these measures to state (apparently new) quantitative uncertainty relations valid in quantum mechanics. We also apply one of such relations, together with the no-signalling assumption, to obtain the Tsirelson's bound in CHSH inequality. We conclude the paper in Section \ref{sec:openPROB}, where we state a number of open problems and directions of further research. We also include Appendices containing proofs of certain technical statements given in the main text.

\section{Framework and notation}\label{sec:framework}
First, we give a survey of notations and concepts used by us in this work.  We will work in the framework of operational theories \cite{SpekkensCont2005,LeiferReview2014}. An operational theory consists consisting of preparations $P$ (belonging to the set $\P$)  and measurements $M$ (belonging to the  set $\M$). An operational theory describes the statistics in a "prepare and measure" scenario, in which a system is  prepared using a preparation procedure $P$ and measured using a measurement device (observable) $M$. Then, the outcome $k$ occurs  with the probability $q_M(k|P)$.  We will use the notation $\q_M (P)\equiv \left(q_M(1|P),\ldots, q_M(n|P) \right)$ to denote the vector of outcome statistics, when a preparation $P$ is measured by a measurement $M$ ( $n$ is the number of outcomes of  $M$). From now on, for the sake of simplicity we will focus on the case of two measurements (observables) $X$, $Y$. A priori in the operational theory $X$ and $Y$ cannot be measured jointly i.e. one does not have access to the joint probability distribution of observing values of \emph{both}  $X$ and $Y$ in a single experiment. Hence, in what follows we will be interested in distributions possible to obtain when measuring \emph{either} of the observables  $X$ or $Y$. Therefore, for a given preparation procedure $P$, the object of interest is then the vector of probability distributions 
\begin{equation}
\q (P) \eqdef (\q_{X} (P)  , \q_{Y} (P) ) \ ,  
\end{equation}
where we dropped the dependence of $\q (P)$ on observables $X,Y$ in order to keep the notation compact. As the preparation $P$ varies  we obtain different probability distributions $\q_{X,Y}(P)$ and consequently different vectors $\q(P)$. We denote the convex set of all allowed vectors $\q (P)$ by $\setth$. The set $\setth$ shall call {\it statistics set for $X$ and $Y$}, or in short {\it statistics set}.
Thus, $\setth$ is embedded in the Cartesian product of two simplices $S\eqdef\Delta_n \times \Delta_n$ (see Fig. \ref{fig:sxy}).
\begin{figure}[http]
\begin{center}
\includegraphics[scale=0.3]{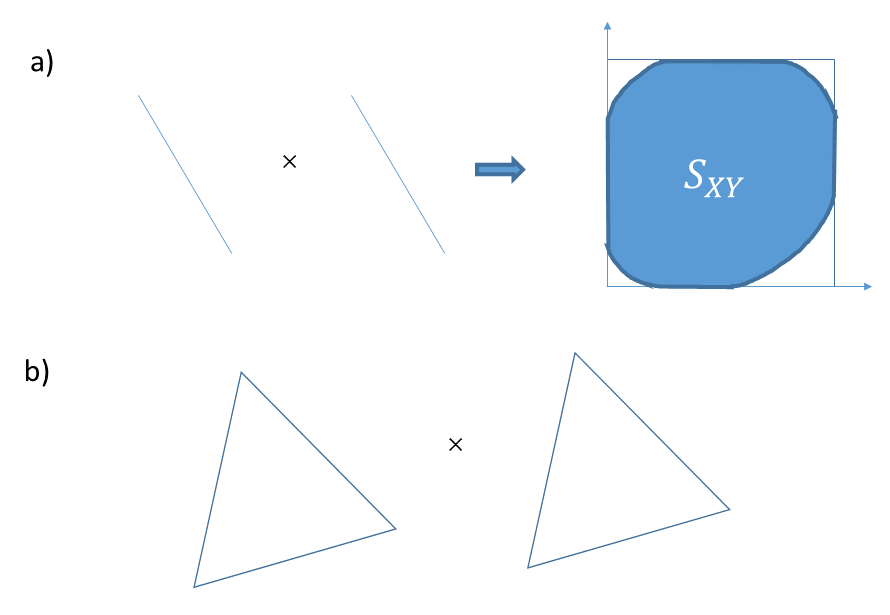}
\caption{\label{fig:sxy} The statistics set is embedded into Cartesian product of two simplices $S=\Delta_n \times \Delta_n$. a) dichotomic observables ($n=2$): 
simplices are one-dimensional  and the axes represent probabilities of a single outcome for each observable.  Their Cartesian product is a square, and $\sxy$ is its convex subset. b) observables with 
three outputs ($n=3)$ give rise to two dimensional simplexes.  Their Cartesian product and the set $\sxy$ cannot be visualized.}
\end{center}
\end{figure}

Note that the set $\P$ can be always assumed to be convex as one can always formally define the mixture of two different preparations via the mixture of the corresponding probability distributions for all measurements $X\in\M$. \mg Operationally this corresponds to choosing between two preparation procedures $P_1$ and $P_2$ by the result of tossing of a biased coin with probability, say, $(\alpha,1-\alpha)$. The output statistics of the resulting preparation $P $ is convex-linear, i.e.,
\be\label{eq:mixpreparations}
\forall X\in \M, \ \q_X(P) =  \alpha  \q_X(P_1) + (1-\alpha)  \q_X(P_2), 
\ee
 and therefore the  statistics set $\setth$ is convex.
Similarly, one can perform convex mixture of two different measurements $X_1,X_2\in \M$ such that the output statistics for all preparation of the resulting measurement $X$ is convex combination of corresponding probability distribution, i.e,
\be\label{eq:mixmeasurements}
\forall P\in \P, \  \q_X(P) =  \alpha  \q_{X_1}(P) + (1-\alpha)  \q_{X_2}(P).
\ee

\blk

\begin{rem*}
Connecting to the standard quantum formalism: in quantum theory preparations $P$ are simply quantum states whereas measurements (observables) $M$ are simply allowed quantum-mechanical measurements.
\end{rem*}

The main aim of this work is to define and study the joint uncertainty \cite{GourNJP2016},  complementarity, uncertainty relations and uncertainty principle in terms of the observed statistics $\q (P)$ and the allowed statistics set  $\setth$.

In what follows we will need a couple more concepts related to classical  manipulation and simulation of observables in general theories. See \cite{Buscemi2005,Haapasalo2012} for the basic definitions in quantum mechanics, \cite{OszmPOVM2017,Oszmaniec2018,Leo2017} for application in quantum information, and a recent work \cite{Filippov2018} for  the extension to the realm of of general probabilistic theories.  

\begin{definition}[Simulation of observables]\label{def:simulation}
We say that observable $X$ can simulate observable $Y$ (denoted as $X \rightarrow Y$), when there exists a stochastic channel $\Lambda$ 
such that if we apply the channel to outputs of the observable $X$, then 
for any preparation $P$, the obtained statistics is the same as  the statistics of the outputs of $Y$ for that preparation. 
\end{definition} 
Formally, $X \rightarrow Y$ 
there exists a stochastic map $\Lambda$ such that  $\q_{Y}(P) = \Lambda \q_{X}(P)$, \emph{simultaneously}, for all preparations $P$. \\  

\begin{definition}[Clean observables]
\label{def:clean}
An observable $X$ is called {\it clean} if for any $Y$ such that $Y\rightarrow X$, also $X\rightarrow Y$.
\end{definition}
In other words, a clean observable is an observable that cannot be simulated in irreducible manner to other observable in the theory. 
\begin{definition}[Sharp observable]
An observable $X$ is called {\it sharp} if for any output there exists a $P$, which gives this output with probability 1.
\end{definition}

\mg
\begin{definition}[Extremal observables]\label{def:extremal}
We say that an observable $X$ is extremal if the statistics of the outputs can not be obtained by convex mixture of two distinct measurements simultaneously for all preparations.
\end{definition} 
\blk

\section{Uncertainty, complementarity  and uncertainty relations} \label{sec:CONCEPT}
\label{sec:compl-unc-general}

\subsection{Preparation uncertainty relation and complementarity}

Let us start with the formal definition of (preparation) uncertainty of two observables $X,Y$.

\begin{definition}[Joint preparation uncertainty]
We say that a preparation $P$ is exhibits joint preparation uncertainty for observables $X,Y$ if at least one of the distributions $\q_X (P),\q_Y (P)$ is not deterministic.  
\end{definition}

In quantum mechanics preparation uncertainty relation (\pur) \cite{BuschHeinsenberg2007} refers to the situation, 
where for two quantum-mechanical observables $X,Y$ there exist 
no praparation (state) $P$ for which both $X$ and $Y$ have well-defined values. Typically, \pur\ has the form
\begin{equation}\label{eq:PUR}
\mathrm{U}_{X,Y}(P) \geq C_{X,Y}\ , 
\end{equation}
where $\mathrm{U}_{X,Y}(P)$ is some measure of joint uncertainty of  $X$ and $Y$ on a preparation  $P$ and $C_{X,Y}$ is the quantity depending on $X$ and $Y$. Often, the right-hand side of \eqref{eq:PUR} is identified with the measure of complementarity of observables $X$ and $Y$. Our goal is to propose a framework allowing to consider the preparation uncertainty relation in any theory. Therefore, both sides of the \pur\ should have operational interpretation i.e. should depend only the observed statistics rather than on  the formalism of the particular theory. 

Currently, in quantum mechanics  the  right-hand side of \pur\ is  typically not defined operationally. Namely, it usually refers explicitly 	to the mathematical structure of quantum mechanics rather than to the observed statistics. For example, in the Kennard-Robertson  uncertainty relation \cite{Kennard1927,Robertson1929} $C_{X,Y}$ depends on the commutator $[X,Y]$. Also, in Deutsch \cite{Deutsch1983}  and Maassen-Uffink \cite{MaasenUfff1988} entropic  \ur\, $C_{X,Y}$ is a function of the maximal overlap of eigenvectors of the involved observables.    Let us note, that in quantum mechanics the right-hand side of \eqref{eq:PUR}  is nontrivial only for noncommuting observables. Such observables have a crucial feature that they {\it access informations that cannot be obtained simultaneously}. In fact, this characteristic has been associated with complementarity already since the invention of quantum theory   \cite{Conmpl1928,Bohr1949,Plotonisky2014,Petz2007}. In this work we propose to define the notion of complementarity of two observables  via impossibility of joint access to informations obtained in the course of their measurements. This allows us to talk about complementarity in any physical theory. Importantly, our definition differs from the approach from \cite{BuschHeinsenberg2007}, where complementarity is defined by the minimal value of uncertainty (the right hand side of \eqref{eq:PUR}) over all states allowed in the theory. This perspective, albeit operational, treats complementarity only as the quantifier of uncertainty of a theory.  Our approach is that complementarity can be regarded as something positive: there is more information in the system than one observable, even most fine grained, can access. This however, at least in quantum mechanics, comes with the price which takes the form {\it uncertainty relations} (of various types that we discuss below). Existence of such price for the phenomenon of excess of information we shall postulate as a physical principle.

\subsection{Complementarity and joint non-measurability}

Let us start with the qualitative definition of complementarity.
\begin{definition}[Complementarity]\label{def:compl}
We shall call two observables $X,Y$ are complementary if they are not jointly measurable i.e. they  statistics $\q_X (P)$, $\q_Y (P)$  cannot be obtained by classical post-processing independent on the preparation $P$. 
\end{definition}
This definition is motivated by the following observation: if two observables are jointly measurable, this means that both informations can be accessed by measuring a single observable. This would mean, that the observables were simply not fine grained enough.  Interestingly, this reasoning, in quantitative form, is  itself an uncertainty relation, called {\it measurement uncertainty relation} (MUR); quoting \cite{BuschLahtiWerner}: "Measurement uncertainty relations are quantitative bounds on the errors in an
approximate joint measurement of two observables". 

In quantum mechanics the two uncertainty relations: MUR and PUR are intimately related.
Namely, PUR can be nontrivial only for those observables for which MUR holds.   Here, we say that PUR is nontrivial, 
if it nontrivially restricts the statistics of the two observables, i.e. that RHS  of \eqref{eq:PUR} is nonzero. 

Let us emphasize here, that  it is not always opposite: namely, even if observables are not jointly measurable
 (i.e. when we have nontrivial MUR),  PUR may be still trivial.  In other words: complementarity not always enforces 
uncertainty. E.g. when we have two observables that have a common eigenstate, but otherwise do not commute, we have 
no joint measurability,  and the observables are still (though not fully) complementary
but PUR is trivial: right hand side of PUR is zero, and there is no uncertainty.  Basic example is given by these observables:
\be
\label{eq:ex1}
\left[ \bea{cc}
\sigma_x & 0 \\
0        & 1 \\
\eea \right], \quad
\left[ \bea{cc}
\sigma_z & 0 \\
0        & 1 \\
\eea \right]\ . \quad
\ee

Interestingly, even more drastic phenomena can happen. Consider two dichotomic projective measurements $M$ and $N$ in $\mathbb{C}^6$ (equipped with the standard basis $\lbrace \ket{i}\rbrace_{i=1}^6$) having the following effects 
\begin{eqnarray}
M_1 &= \ketbra{1}{1}+\ketbra{3}{3}+\ketbra{4}{4}\ \nonumber ,\\
M_2 &= \ketbra{2}{2}+\ketbra{5}{5}+\ketbra{6}{6}\ \nonumber ,\\
\label{eq:weird}
N_1 &= \ketbra{+}{+}+\ketbra{3}{3}+\ketbra{5}{5}\  ,\\
N_2 &= \ketbra{-}{-}+\ketbra{4}{4}+\ketbra{6}{6}\ \nonumber,
\end{eqnarray}  
where $\ket{\pm}=(1/\sqrt{2})(\ket{0}\pm\ket{1})$. It can be seen that even though the above measurements are not jointly measurable (because the states $\ketbra{\pm}{\pm}$ do not commute with the states $\ketbra{0}{0},\ketbra{1}{1}$) , there is \emph{no uncertainty} - in fact the statistics set $S_{N,M}$ is as big as possible and equals $S$, the Cartesian product of two one dimensional simplices (see Fig.\ref{fig:sxy}). Notice however that the above projective measurements are not (see Definition \ref{def:clean}) since they can be obtained as coarse-grainings of fine-grained (rank-one) projective measurements in $\mathbb{C}^6$. In what follows we will show that in quantum mechanics (suitably-understood) joint non-measurability indeed implies (suitably-understood) uncertainty, but only for clean and \mg extremal \blk observables.

\subsection{Three types of uncertainty and complementarity}\label{sec:3comp}

The above discussion shows that joint non-measurability may seem to be not a good candidate for right hand side of \eqref{eq:PUR}.
Fortunately, there is an extension of PUR, called exclusion principle proposed by Hall \cite{Hall}. While the original 
Hall's principle,  is still trivial for observables that share a common eigenstate, its natural extension conjectured in \cite{GrudkaExclusion}
and proved in \cite{ColesPiani}, is nontrivial, whenever observables do not commute.   The exclusion principles are quantified in particular manner (via mutual information). We would like to avoid using any particular quantifiers as at the moment we are only interested in the question, of whether there is uncertainty, or not, and whether there is information exclusion or not. In what follows we present the qualitative definitions of uncertainty and exclusivity that avoid usage of any quantifiers.

\begin{definition}[Traditional uncertainty]\label{def:unc}
Two observables $X,Y$ exhibit non-zero (preparation) uncertainty, if  for arbitrary preparation $P$, their statistics $\q_X (P)$, $\q_Y (P)$ are never both deterministic at the same time.
\end{definition}

\begin{definition}[Information exclusion]\label{def:excl}
Consider two observables $X,Y$ with $d$ outcomes. We say that they have {\it information exclusion}, if there does not exists $d$ element set of preparations $P_i$, so that each of the states gives fully predictable output for both observable, and different state leads to a different outputs (statistics). 
\end{definition}

From now on we can operate solely  on a qualitative level. In quantum mechanics, whenever sharp  and clean  measurements (i.e. projective measurements with  one-dimensional projections)  are not jointly measurable (equivalently, they do not commute \cite{Heinosaari2008}), they lead to nontrivial information exclusion principle,  
ergo complementarity of two observables always imposes nontrivial exclusion principle on those observables.  
For formal proof see Lemma \ref{lem:appendix-implications} in Appendix \ref{sec:appendix-implications}.
Recall that for non-clean observables, it is not true, as shown by measurements given in Eq.\eqref{eq:weird}. \mg Note that in quantum theory sharp and clean observables are extremal too. \blk 
Thus, we have the following: {\it In quantum mechanics for clean observables  complementarity implies information exclusion}. 

As said, we cannot replace in this sentence "information exclusion" with "uncertainty". 
Thus we obtained a picture illustrated by Table \ref{table:uc1}, where we have one space to
fill: some version of complementarity, that would imply traditional uncertainty. 
\begin{table}[h]
\begin{tabular}{||l|c||}
\hline
\hline
Uncertainty & Complementarity \\ [0.2ex]
\hline\hline
information exclusion & associated with  \\ 
 & joint non-measurability \\ 
\hline
traditional & ? \\
\hline\hline
\end{tabular}
\caption{\label{table:uc1}}
\end{table}
Now we would like to fill it. 
 Let us note that if we coarse grain the observables from the example given in  Eq.\eqref{eq:ex1},
by choosing not to distinguish between the two outcomes of $\sigma_x$ and the same $\sigma_z$, 
then the new observables will become trivial, having no complementarity and no uncertainty. 
This prompts us to consider a stronger version of complementarity, which can be called  \emph{full complementarity}.
\begin{definition}[Full complementarity]
We say that two observables $X,Y$ are fully complementary when after arbitrary 
coarse-graining (apart from the trivial one, where none outcomes are not distinguished) the observables still remain jointly not-measurable. 
\end{definition}
Clearly, such stronger complementarity \emph{implies uncertainty in the traditional form for projective measurements} (it follows
from Lemma \ref{lem:appendix-implications} in Appendix \ref{sec:appendix-implications}). 
However, let us consider the following example
\be
\label{eq:ex2}
\left[ \bea{cc}
\sigma_x & 0 \\
0        & \sigma_x \\
\eea \right], \quad
\left[ \bea{cc}
\sigma_z & 0 \\
0        & \sigma_z \\
\eea \right]\ . \quad
\ee
The above two observables do not exhibit full complementarity, yet they are uncertain. Thus, this notion is a bit too strong to be put in 
the table on the same level as traditional uncertainty. At a first glance, such strong notion of complementarity 
should be associated with the following strong version of uncertainty, which, to our knowledge has not been examined so far. 
\begin{definition}[Strong preparation uncertainty]
We say that two observables $X,Y$ exhibit strong (preparation) uncertainty when they remain uncertain after any nontrivial coarse-graining. In other words, it is impossible to find a preparation $P$ such that $\sum_{i\in I} q_X(i|P) =\sum_{j\in J} p_Y(j|P)=1$, for some nontrivial subsets $I,J$ of the output spaces of $X$ and $Y$ respectively.
\end{definition}

\begin{rem*}
It is also possible to define a strong information exclusion. Namely, we say that observables $X,Y$ exhibit strong exclusion, when after any coarse-graining  they still exhibit information exclusion. Interestingly, in quantum mechanics, the two notions become equivalent, however  in general (for some weird theory) they may be distinct.
\end{rem*}

Somehow counter-intuitively, it turns out that in quantum mechanics full complementarity does not imply full uncertainty, even for clean and \mg extremal \blk measurements (see Appendix \ref{app:fullCOMP} for the concrete counterexample in dimension five).  Therefore, in quantum mechanics full complementarity and strong uncertainty will not give rise to uncertainty like principle. In turns out that the version  of complementarity that implies traditional uncertainty (for clean observables) is the following intermediate version of complementarity, which we shall call \emph{single-outcome} complementarity. The proof is given in Lemma \ref{lem:appendix-implications} in Appendix \ref{sec:appendix-implications}. 


\begin{definition}[Single-outcome complementarity]
We say that two $d$ outcome observables $X,Y$ exhibit single-outcome complementarity when after coarse-grainings, that preserve one outcome, and glue all the rest $d-1$ outcomes, the resulting dichotomic observables are still jointly non-measurable. 
\end{definition}

\begin{table}[h]
\begin{tabular}{||l|c||}
\hline
\hline
Uncertainty & Complementarity \\ [0.2ex]
\hline\hline
information exclusion & associated with  \\ 
 & joint non-measurability \\ 
\hline
traditional & associated with \\
						& single-outcome joint non-measurability \\
\hline\hline
\end{tabular}
\caption{\label{table:uc3}}
\end{table}
Summarizing, for quantum mechanics we have obtained the full picture, as shown in Table \ref{table:uc3}.

\subsection{Uncertainty principle as a physical postulate}

Motivated by the analysis presented in the preceding part, we have found candidates for the right hand sides of the general uncertainty relation  \eqref{eq:PUR}. These will be one of variants of complementarity, depending on what type of uncertainty we will put to the left hand side. The implications between our notions both those that hold by definition, as well as those postulated as (qualitative) uncertainty relations are depicted in  Fig. \ref{fig:implications}. We also show in the figure the pair strong uncertainty vs full complementarity, 
pointing out that the implication does not hold. 
\begin{figure}[h]
\includegraphics[scale=0.5]{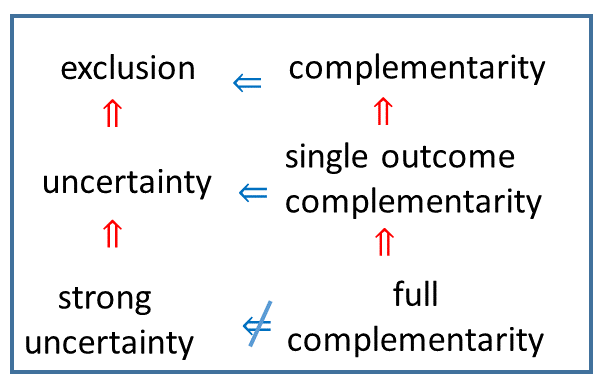}
\caption{\label{fig:implications} The red implications hold by definition. The blue ones we postulate for  clean, sharp and 
\mg extremal \blk observables in  physical theories.}
\end{figure} Recall that in quantum mechanics relations between the two kinds of uncertainty and complementarity given in 
Fig. \ref{fig:implications} hold only for fine grained projective measurements. These measurements are clean \mg and extremal \blk observables (see Definition \ref{def:clean}) and we postulate the relation between uncertainty and complementarity only for \mg clean-extremal \blk measurements.  

{\bf Postulate (Uncertainty principle):} {\it In physical theories \mg observables which are complementary, clean and extremal, \blk necessarily exhibit uncertainty.}

In other words: in any theory lack of joint measurability for \mg clean-extremal \blk observables must imply uncertainty. In other words the existence of uncertainty principle can be also understood as a price for the excess of information provided by complementary observables: 
{\it Uncertainty principle states that complementarity has a price -  which is uncertainty.}

\begin{rem*}
Let us emphasize, that while uncertainty is present only in quantum world, and 
not in classical one, the uncertainty principle holds both in quantum and classical theory:  In classical case it holds,
because there is no complementarity, and therefore the "price" is zero.
\end{rem*}

\begin{rem*}
In this work we will be mostly interested in sharp and \mg extremal \blk observables, as  non-sharp \mg or non-extremal \blk observables are themselves uncertain, and the uncertainty is not related to complementarity, but just comes from some form of apriori epistemic restrictions. Note however that the existence of non-sharp \mg or non-extremal \blk measurements \emph{does not} contradict the uncertainty principle.
\end{rem*}

Later in this paper we shall pave the way to quantify uncertainty and complementarity, 
aiming to grasp the above principle quantitatively.  At this moment let us informally state the general form of uncertainty relations. 
\begin{definition}[General uncertainty relation]
The general uncertainty relation is inequality of the following form
\begin{equation}\label{eq:pur}
\U_{X,Y}(P) \geq f^{\uparrow}(C_{X,Y})\ ,
\end{equation}
where $\U_{X,Y}(P)$ is a measure of (joint) uncertainty of $X$ and $Y$, $C_{X,Y}$ is some indicator of complementarity of observables
$X,Y$ (see  Section \ref{sec:post-uncert} for the properties that these quantities should satisfy ), and  $f^{\uparrow}$ is a non-decreasing functions whose ranges are non-negative. The form of these two functions depends on the particular measures of  complementarity and uncertainty used.
\end{definition}

\blk

So far we have mostly talked about the negative aspect of complementarity (joint non-measurability), however as we have mentioned,
it is strictly connected with a positive aspect of complementarity: because of joint non-measurability, the observables reveal more information, than possible by means of a single observable. Further in the paper we will provide examples of quantifiers of uncertainty that would reflect this point of view.

\begin{rem*}Let us emphasize, that the notion of complementarity we propose differs from the one considered in \cite{OW}:
"(...) two measurements are complementary, if the second measurement can extract no more information about the preparation procedure than the
first measurement and visa versa. We refer to this as information complementarity. Note that quantum mechanically, this does not necessarily have to do with whether two measurements commute. For example, if the first measurement is a complete Von Neumann measurements,
then all subsequent measurements gain no new information than the first one whether they commute or otherwise."
We see that the authors consider \emph{sequential} measurements, and that their definition incorporates the process of disturbing 
the state by measurement. In our paper we restrict to the typical scenario of preparation uncertainty relations, where there are no sequential measurements, and our complementarity is built-in in such a paradigm. 
\end{rem*}
\mg
\subsection{Reverse uncertainty relations}
\label{subsec:reverse}
One can also ask, how about inverse relation, where complementarity would imply uncertainty. 
We may consider the following definition: 
\begin{definition}
Reverse PUR is the following implication: non-zero uncertainty implies non-zero complementarity. I.e. uncertainty cannot occur if observables are not complementary to some extent. 
\end{definition}
Note that, while uncertainty principle does not hold in arbitrary theory, and we want to propose it to be a postulate for legitimate theories,
the above {\it reverse PUR} is expected to  hold for all sharp and clean pairs observables. 
In sec. \ref{subsec:reverse-rescaling} we present result which says, that reverse PUR holds for binary, sharp and clean  outcomes 
for any theory. We give there quantitative form of such reverse PUR. 
In quantum mechanics, it is easy to see, that reverse PUR holds qualitatively for the pair exclusion-complementarity: i.e. exclusion implies complementarity. 
\blk

\section{Complementarity from statistics set and Independence}
\label{sec:basic-indep}
In the previous section, while discussing how to make uncertainty relations operational, we have put emphasis on
connection between  complementarity and  impossibility of joint measurement. Yet, one should also embrace the positive aspect of complementarity: it is the surplus  of information provided by two (or perhaps more) observables. In this section we would like to describe how one can quantify such excess  in arbitrary theory.

Consider a very simple theory: it has just two dichotomic observables $X$ and $Y$, and all possible pairs of distributions are allowed (i.e. 
for any pair of distributions there exists preparation, that gives rise to these distributions, via measurement of our observables.
The statistics set $\sxy$ is therefore the full square (see Fig. \ref{fig:sxy-basic-sec}).  Clearly, each of them brings completely independent information,  and  these two informations cannot be acquired in any other way. 
Thus the two observables are maximally complementary. 
\begin{figure}[http]
\includegraphics[scale=0.3]{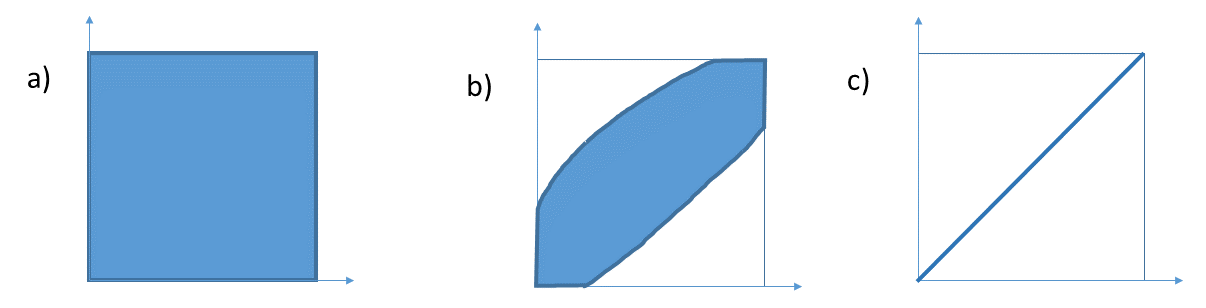}
\caption{\label{fig:sxy-basic-sec} The statistics set for a) most independent observables b) intermediate case c) the same observables}
\end{figure}

Suppose that the set $\sxy$ shrinks a bit towards one of the diagonals. The observables become correlated, although there is no joint distribution.  Namely, measuring any of them does not bring a lot of new information, compared to the information already provided by the measurement of the other one.  This is clearly visible in the extreme example, when the set $\sxy$ is just the diagonal and the observables are identical. Thus, \emph{the more the set shrinks, the smaller is complementary}. Since our two observables are the only ones in the theory, the complementarity is solely a function of the statistics set $\sxy$ i.e. $C_{X,Y}= C(S_{X,Y})$. Moreover, it should be intuitively monotonic under inclusions, i.e. if $\sxy\subset S_{X'Y'}$ then  $C(S_{X,Y})\leq C(S_{X',Y'})$. To summarize: if $X$ and $Y$ are the only observables in the theory, complementarity can be identified with their "independence", which can be intuitively deduced from the statistics set.

The problem becomes more complicated when there are other observables in the theory. To see it, 
consider a quite opposite situation - two classical bits. $X$ measures one bit, and $Y$ measures the other. The set $\sxy$ 
is the same - again square. But complementarity vanishes, as the information can be accessed by refined observable with four outcomes - 
the two bit observable. Thus, for observables that are not clean  the statistics set does not tell us anything about complementarity. 

\mg Similarly, the statistics set of non-extremal observables does not capture \com. Suppose two observables $X_1$ and $X_2$ are not complementary with observable $Y$ separately. We naturally expect that the complementarity between $Y$ and another observable $X$ which is realized by some convex mixture of observables $X_1,X_2,$ is also zero. However,  in general ''independence'' does not satisfy this feature. We provide an example in Appendix \ref{appendix:non-extremal}. \blk

Therefore, in what follows we limit ourselves to  clean \mg and extremal \blk observables. We can now and ask again, whether independence $\prec(S_{X,Y})$ (for a while  intuitively defined function of the statistics-set $S_{X,Y}$, as elaborated above) 
is related to complementarity (joint non-measurability). Or more concretely - can we infer complementarity looking solely at statistics set for two clean \mg and extremal \blk observables?
By definition, for clean \mg and extremal \blk observables there does not exist any \mg set of observables \blk that might reproduce two observables exactly. If one observable can simulate the other one (see Definition \ref{def:simulation}) the statistics set has zero measure. Hence, if the set $\sxy$ a bit thicker than just the diagonal, this must imply that we have complementarity. 

However, quantitatively we might still have the following situation:
there exists third observable, that almost simulate our observables $X$ and $Y$. And this observable would be able to acquire almost all 
the information, hence the independence of the observables would again mean just standard independence, and would not imply complementarity. 
In  such a theory, the (approximate) joint measurability is not revealed in the statistics set. Note 
that in quantum mechanics it is not so. Consider e.g. qubit observables. When they are complementary, 
the set is circle. When they become more and more similar (ergo better and better jointly measurable) the statistics set shrinks 
towards diagonal (see Fig. \ref{fig:quantum-basic-sec}). 
\begin{figure}[http]
\includegraphics[scale=0.3]{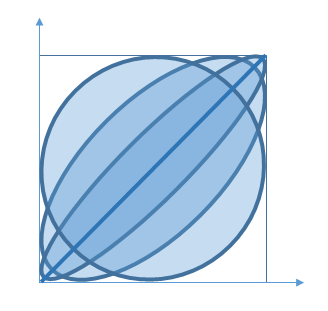}
\caption{\label{fig:quantum-basic-sec} The statistics sets for quantum binary observables. Circle is  for most complementary (e.g. $\sigma_x$ and $\sigma_z$, diagonal for both being $\sigma_z$).}
\end{figure}

To summarize: for clean \mg and extremal \blk observables independence may not reflect complementarity, in a theory, where  better and better joint measurability of clean observables does not imply that the observables converge to one another. 
Thus in general one should somehow connect two features: (i) how well observables can be simulated by a third one (which is 
a subject of  MUR) (ii) independence seen in statistics set. And complementarity would be a function of those two features. This looks like a very ambitious  program, and therefore for the purpose of this paper, we shall take a first step. 
Namely, we shall work out complementarity, that will work well in theories where approximate joint measurability (for clean \mg and extremal \blk observables) means that the observables are approximately the same.  Thus, in the rest of the paper, we will assume that the statistics set of clean and sharp observables properly reflect the joint measurability features.

Finally, we can define the complementarity through independence for arbitrary \mg extremal \blk observables as follows, 
\be\label{eq:compFROMind}
C_{X,Y} \eqdef \min_{X':X'\rightarrow X} \min_{Y':Y'\rightarrow Y} \prec(S_{X',Y'})\ ,
\ee
where the minimum is taken over all observables $X'$ and $Y'$ that simulate $X$ and $Y$ respectively. \mg
For non-extremal observables, we follow the convex-roof extension of the above definition, that is,
\be \label{com-non-ext0}
\C_{X,Y} = \min_{\{\alpha_i,X_i\}} \min_{\{\beta_j,Y_j\}} \sum_{i,j} \alpha_i \beta_j\ \C_{X_i,Y_j}\ ,
\ee
where the minimum is taken over all possible decomposition of the observables $X,Y$  to the extremal observables $\{X_i\},\{Y_j\}$ with probability distribution $\{\alpha_i\},\{\beta_j\}$. 
 Importantly, this notion of complementarity reduces to independence for clean and extremal observables. \blk

\section{Dichotomic observables - intuitive picture} \label{sec:twoOUTCOMES}

In this part we focus exclusively on the case of dichotomic observables. This simplified setting allows for  the appealing geometrical interpretations of the ideas presented in the preceding sections.
As mentioned before, for two observables $X,Y$, each with two outputs the simplices are just intervals, and the product of two simplices is a square. 
The set of $S_{X,Y}$ is some convex body within the square. Possible sets $\sxy$ are depicted in Fig. \ref{fig:2out-sxy}.
\begin{figure}[http]
\includegraphics[scale=0.3]{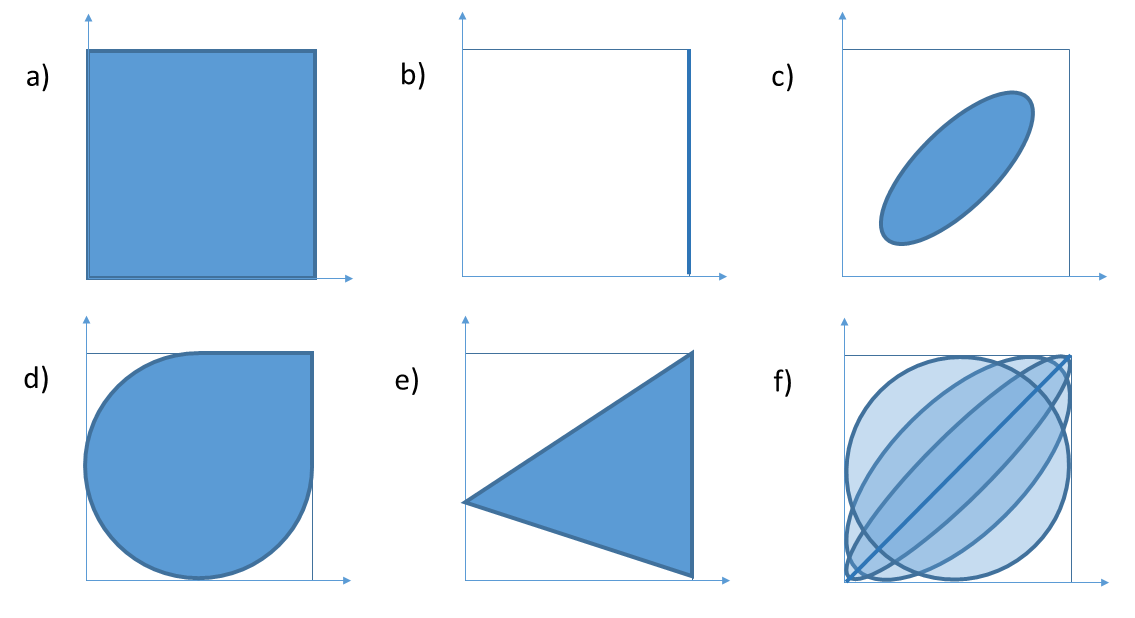}
\caption{\label{fig:2out-sxy} Various sets $\sxy$. (a) $\sxy$ is equal to full square - the so called "square bit". (b) one observable is completely noisy - reports no information (c) both observables are not sharp i.e. there is no state that would give deterministic outcome for any of them.(d), (e) both observables are sharp. (f) quantum mechanical observables.}
\end{figure}

If both observables are sharp, i.e. for any outcome there exists a state, that gives this outcome with probability $1$, 
the set must touch each of the edges of the square.  The examples of non-sharp observables are in Fig. \ref{fig:2out-sxy}b) and c).  In Fig. \ref{fig:2out-sxy}f) we have qubit observables of the form  $X={\bf n}\cdot \sigma$, $Z=\sigma_z$,  with $n_y=0$, and $n_x^2 + n_z^2=1$.
Depending on angle between the vectors $\bf n$ and $(0,0,1)$, we interpolate between 
(i) the classical case, where both observables are $\sigma_z$, and the set $\sxy$ is just a line connecting opposite corners, 
 and (ii) most complementary case, where the set $S_{X,Y}$ constitutes a circle, and observables are $\sigma_x$ and $\sigma_z$
in the latter case two observables are "mutually unbiased", i.e. for any state that gives deterministic outcome for one observable,
it gives completely random output.


\subsection{Independence/Complementarity}

Note, that for two outcomes, there is no distinction between the three kinds of complementarity/independence presented in Section \ref{sec:3comp}. This is because there is no non-trivial coarse-graining operations. Assuming that observables are clean \mg and extremal \blk, we can now identify complementarity and independence (see discussion in Section \ref{sec:basic-indep}).

{\it Square bit:}  
For the states to be corners, both observables bring maximal, and independent information. 
Clearly the square presents the richest statistics that can be obtained  from two observables, therefore it has the largest possible independence among all sets $\sxy$. 

{\it Classical bit:} The set $S_{X,Y}$  is just diagonal or anti-diagonal. In the first case the second observable is just a copy of the  first one, and in the second case - its negation. Here both observables report exactly the same information. Ergo, we have 
no independence.

{\it Qubit:} For observables $X={\bf n}\cdot \sigma$, $Z=\sigma_z$,  with $n_y=0$, and $n_x^2 + n_z^2=1$, see Fig. \ref{fig:2out-sxy}f)
depending on angle between the vectors $\bf n$ and $(0,0,1)$, we interpolate between 
the classical case, where observables are the same, and the most complementary case possible in quantum mechanics, where the set $S_{X,Y}$ constitutes a circle.
This latter is the case, where two observables are "mutually unbiased", i.e. for any state that gives deterministic outcome for one observable,
it gives completely random output. Note that this randomness is not a signature of complementarity.   Exactly the same behavior occurs  also for the square bit, where we can have states deterministic for both observables. Rather it should be regarded as  {\it uncertainty}.

Generally, for dichotomic clean \mg and extremal \blk observables, whenever the statistics set is thick (i.e. not one-dimensional) we expect nonzero complementarity. 
In particular, the measures that we shall propose further, in the case of two outcomes will all have  this feature.



\subsection{Uncertainty}

\begin{figure}[http]
\includegraphics[scale=0.3]{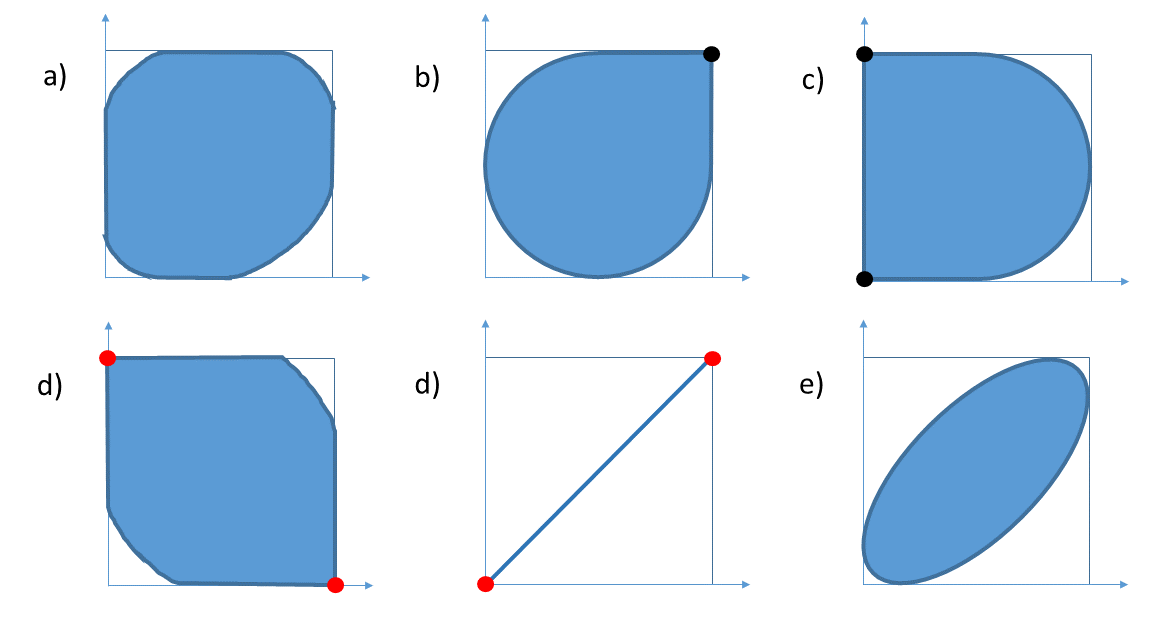}
\caption{\label{fig:2out-uncert} Uncertainty for two outcome observables.  a) no corner included, hence we have uncertainty for any state 
b) one corner included -  represents preparation that has no uncertainty for both observables; exclusion still holds 
c) two corners included, so for two preparations no uncertainty, still exclusion holds  d) no uncertainty and no exclusion, since opposite corners are included e) classical case (the same observables) - no uncertainty f) generic quantum observables: both uncertainty ad exclusion.}
\end{figure}

The concept of uncertainty for dichotomic observables is illustrated in Fig. \ref{fig:2out-uncert}. The only preparations, which give deterministic statistics for both observables correspond to corners of the square. 
The traditional uncertainty thus means that the set $S_{X,Y}$ does not include any corner. Exclusion means that the set 
does not include any pair of opposite corners. Thus, unlike in the case of complementarity, even for two outcomes, 
uncertainty does not reduce to one type: there can be situation, that exclusion holds, but there is no uncertainty, see Fig. \ref{fig:2out-uncert}b) and \ref{fig:2out-uncert}c).  Clearly, strong uncertainty and traditional uncertainty collapse into one notion, since 
there is not nontrivial coarse graining for two outputs.  Thus we are left with two types of uncertainty. 
Note, that in quantum mechanics for two outcomes, at least qualitatively, there is no difference between 
the traditional uncertainty and exclusion. 

Finally, note that in \cite{SteegWehner} theories were considered, whose elementary systems exhibit
the statistics set $\sxy$ described by equation: 
\begin{equation}
\left(\frac{x-1}{2}\right)^p + \left(\frac{y-1}{2}\right)^p \leq 1
\end{equation}
for $p\geq 1$. 
For $p=2$ it is circle, i.e. the quantum case of maximally complementary observables (i.e. circle). 
For $p\to \infty$ the set $\sxy$ tends to full square.

\subsection{Uncertainty principle} 

As said in Section \ref{sec:compl-unc-general} and \ref{sec:basic-indep}, preparation uncertainty principle 
says that there is a price for complementarity: 
namely complementary observable have to be uncertain. 

For two outcomes, uncertainty principle says that whenever complementarity is nonzero, e.g. when the statistics set $\sxy$ is not one dimensional, then the set does not contain corner.  On more quantitative level, uncertainty principle says that that the more 
complementarity we want, the larger must be uncertainty.  
We see this in quantum case: the more we  want to be close to all four corners, the more we depart from the two original corners, 
which belonged to $S_{X,Y}$ in the case of classical bit (i.e. when two observables were the same). We observe this 
in Fig. \ref{fig:2out-ur-a}, where  we show sets $S_{X,Y}$ for three values of the angle between observables.

\begin{figure}[http]
\includegraphics[scale=0.3]{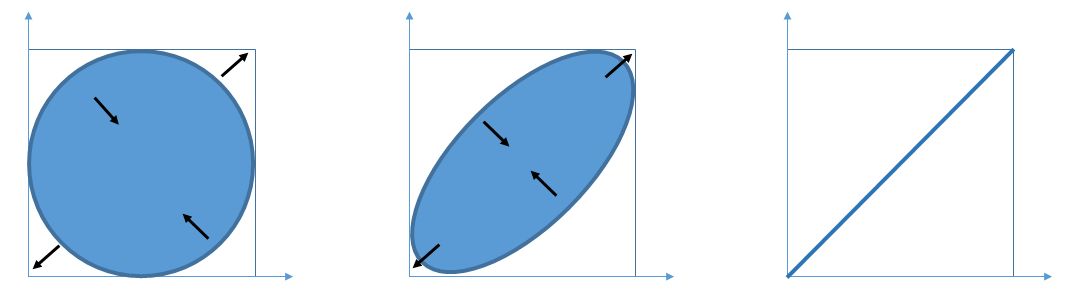}
\caption{\label{fig:2out-ur-a} Uncertainty principle in quantum case.
}
\end{figure}

Generally, uncertainty means that the set $S_{X,Y}$ is far from any of the corners. Complementarity means that $S_{XY}$ is 
close to all the corners.  Thus, uncertainty principle says: when one wants to be close to any one of the two opposite corners, one cannot 
be close to the other opposite corner.  Thus uncertainty principle puts also bounds on complementarity itself: 
the maximal complementarity can be achieved only when uncertainty vanishes, but this is forbidden by uncertainty principle.


\section{Outline of the further results of the paper}\label{sec:outline2}

In this section we will give motivation and overview of the results presented in the second half of the manuscript.

\subsection{Quantifying independence and complementarity and proposing uncertainty relations.}
In the paper we shall propose some postulates that measures of uncertainty (Section \ref{sec:post-uncert}). They are just modest 
updating of the postulates given in \cite{FriedlandGour2013,GourNJP2016}).
Then we propose postulates for measures of independence and complementarity in Section \ref{sec:post-compl}.

Subsequently we shall propose some concrete measures of complementarity. Mostly we will 
concentrate on one of the types out of three presented in Section \ref{sec:compl-unc-general}: 
the most basic one that does not involve coarse graining.
We shall propose measure by means of random access codes in Section \ref{sec:random-access}, by means of rescaling in Section \ref{sec:rescaling}, and by means of preimages  in  Section \ref{sec:preimage}.
{\it A priori} we might not be able to make from these proper PUR, because in table \ref{table:uc3} they are in different rows. 
However, as already discussed in Section \ref{sec:basic-indep} for binary outcomes all complementarities coincide. We shall also propose uncertainty
based on random access code in Section \ref{sec:random-access}.
 

Now, having more or less compatible candidates for uncertainty and complementarity, one would like to build uncertainty relation,
that might be imposed on all theories. Let us emphasize, that we do not necessarily want the simple form of Eq. \eqref{eq:PUR}.
We will be satisfied with any relation, that will constrain uncertainty by complementarity. 

One way of obtaining uncertainty relations to be imposed on physical theories is to find what a relation between proposed uncertainty $U$ and 
complementarity $C$ is satisfied in quantum mechanics. An example of such PUR will be the relation \eqref{eq:urq2} between $C$ and $U$ built on the basis of rescaling.

\subsection{Relation with Information Contents Principle}
Having proposed some understanding of what uncertainty principle can mean in operational terms, would be good to have a 
universal PUR that is not forcefully build to fit quantum mechanics. An example of a principle that holds in quantum mechanics 
even though was not deliberately chosen to do so is Information Causality \cite{ic}. In \cite{Czekaj-ICP} a version 
of Information Causality was proposed that  differs mainly by putting emphasis on a single system,
while Information Causality apriori deals with bipartite systems.  It was called Information Content Principle (ICP).
It represents a bound on  random access codes for ensembles of states quantified by the mutual information. 
Therefore, qualitatively, it prevents maximal complementarity (if the latter is expressed by means of random access code). 
In Section \ref{sec:icp} we will show, that if the set $S_{XY}$ is symmetric under rotation about $\pi/4$ 
as is in the case of quantum mutually unbiased observables, then ICP turns out to be Maassen-Uffink 
uncertainty relation for such observables.  
We also show, that even with less symmetry assumptions, it still provides constraints for $\sxy$ which can play a role of PUR,
namely ICP prevents from too much complementarity, if there is not much uncertainty.

\subsection{Consequences of uncertainty relation for nonlocality}
One of the interesting applications of the idea of operational uncertainty relations, 
which we will present in Section \ref{sec:nonlocality}, is that 
they can put bounds on nonlocality. 
It is ubiquitous problem of quantum information theory, to understand 
in operational terms, what prevents Quantum Mechanics to be less non-local than it would be possible if the only constraint would be no-signaling, see e.g. \cite{BrassardPrinciple2006,ic,NavascuesWMacro}  (in \cite{Lodyga-MUR} the opposite direction was explored too: nonolocality and no-signaling implies measurement uncertainty). In 
Oppenheim and Wehner \cite{OW} attempt to 
understand why quantum mechanics is not maximally non-local, namely, they have made a crucial observation  that the system  that exhibit maximally non-local behavior,  i.e. it violates CHSH  inequality up to its algebraic bound, exhibits no uncertainty. 

Indeed, consider CHSH inequality. Alice and Bob measure one of two observables $A_1$, $A_2$ and $B_1$, $B_2$. 
When Alice measures her observable $A_1$, and gets some outcome, she prepares the state on Bob's site. 
To maximize CHSH, Alice's outcome should be perfectly correlated with Bob outcome, for any of his two observables. 
Thus the state of Bob's system, prepared by Alice measurement and outcome, must give deterministic answer to both his observables. 
Thus his observables cannot exhibit uncertainty.  This suggests that it is uncertainty that bounds 
the non-locality. However, there is a problem here: classical systems do not exhibit uncertainty, and yet still are not maximally nonlocal, 
even more - they are  not non-local at all.  Thus saying that uncertainty put bounds on non-locality would be a very weak statement - as it would  not provide any bound on nonlocality of classical systems, and in consequence could not capture the phenomenon 
of non-maximal nonlocality of quantum mechanics. 

The way out proposed in \cite{OW} was to involve also steering. 
 To quote the authors:  "(...) the degree of non-locality of any theory is determined by two factors – the strength of the uncertainty principle, and the strength of a property called “steering”, which determines which states can be prepared at one location given a 
measurement at another. (...)  For any physical theory we can thus consider the strength of non-local correlations to be a
tradeoff between two aspects: steerability and uncertainty.". Some disadvantage of  this approach is that it cannot be based only on statistics of observables in question. 
To verify the statement, the authors had first to find observables that are optimal for violation of Bell inequality, and then for those observables optimize steering. 

Here, we propose a different way out, possible to spell out in operational terms. Namely we just add to the word "uncertainty" 
just another word "principle", i.e. we say: 
"Uncertainty {\it principle} puts bounds on non-locality". Since, as discussed above, 
uncertainty principle holds for the whole quantum theory (unlike uncertainty, which appears only for specific observables), our statement implies also bounds on nonlocality for classical systems. We thus arrived at the following explanation, why quantum theory 
is not maximally non-local:

{\it Quantum theory is not maximally non-local because of uncertainty principle.}

Note that in \cite{OW} some stronger claim was made: namely, that uncertainty and steering not only bounds the non-locality, 
but it actually {\it determines} its value. 
This was later refuted in \cite{RaviGMH-ur-nonlocality}. However, the weaker statement that uncertainty and steerability properties limit non-locality is still meaningful. Also in our case, we are on the same level: we claim that uncertainty principle puts bounds on non-locality.

Here we will argue, how uncertainty principle bounds non-locality 
for clean \mg and extremal \blk observables on a qualitative level. In Section \ref{sec:nonlocality} we shall provide quantitative picture, reproducing Tsirelson bound. 
For binary outputs,  notions of complementarity discussed in Section \ref{sec:compl-unc-general}
all become the same. Thus, uncertainty principle means qualitatively that complementarity implies uncertainty of any of three kinds.
Now, for binary outcomes uncertainty means, that the set $\sxy$ does not include any corner. 
Indeed if a corner belongs to $\sxy$, this means that there exists preparation, such that both distributions are deterministic.  In  Section \ref{sec:nonlocality} we shall argue, that from no-signaling it follows 
that to have maximal violation of CHSH one needs two observables with set $\sxy$ being square. One can see it quickly in the following way: 
to violate CHSH maximally, one needs so called Popescu Rochrlich box. From its very definition it follows, that after Alice's measurement, 
she prepares such states on Bob's side, that all four corners appear. 

Now, we employ uncertainty relation: since Bob's observables will have $\sxy$ being square, then complementarity is nonzero. 
However, uncertainty principle says that then there must be uncertainty, i.e. the set cannot touch corners, and therefore cannot be a square. 
In short, uncertainty principle rules out square, and therefore CHSH cannot be maximally violated. A drawback of our approach is that it works only for clean \mg and extremal \blk observables. Observables that are not clean, can
have $\sxy$ to be square, without uncertainty  - e.g. if one observable is one bit and the other is the other bit on the total system of two bits.

\blk

\section{Postulates for measures of uncertainty} 
\label{sec:post-uncert}

In this part we give the postulates for measures of  uncertainty for two observables. 
We shall not take the order from the weakest to the strongest (which would be: exclusion, traditional uncertainty, strong uncertainty/exclusion)
. Instead, we will begin with the most well known - uncertainty. Then, we will proceed with its immediate derivative - strong uncertainty, and end up with 
exclusion, which is the most complicated one.

\subsection{Uncertainty} 
First, any measure of the joint uncertainty $\U$ of two observables (measurements) $X$ and $Y$ should depend on the observed statistics in particular preparation procedure i.e. we should have $\U(P)= \U(\q(P))$. Intuitively, the measure $\U$ should tell us to what extent it is impossible to have simultaneous knowledge about both $X$ and $Y$ for a given preparation $P$. We propose the following postulates for the measure of of joint uncertainty (note that they are closely related to the postulates given in \cite{FriedlandGour2013,GourNJP2016}) .

\begin{enumerate}
\item We assume $\U(\q(P))\geq 0$ and $\U(\q(P))=0$ if and only if distribution of $X$ \emph{and} $Y$ giving rise to $\q(P)$ are deterministic. 
In other words $\q(P)$, is not located in the corner of the cartesian product of two simplices, see Fig.\ref{fig:2out-uncert}.
\item We assume that $\U(\q(P))$ measure cannot decrease under doubly stochastic operations performed \emph{independently} on outcomes of observables  i.e. 
\begin{equation}\label{eq:indepStoch}
\U\left(( D_1,D_2) \q(P)\right) \geq \U(\q(P)) ,
\end{equation}
for all doubly-stochastic $n\times n$ matrices $D_1$ and $D_2$. 
\item $\U(\q(P))$ measure cannot increase under coarse-graining and permutations of outcomes. Formally,
\be
\label{eq:coarse}
\U\left((\Lambda^e_1,\Lambda^e_2) \q(P)\right) \leq \U(\q(P)) \ee
for any \textit{extremal} stochastic maps $\Lambda^e_{1,2}$. 
\item We assume that $\U(\q(P))$ cannot decrease under taking mixture of preparations i.e. $\U$  is concave with respect to the convex structure of preparations
\begin{equation} \label{concavityU}
\U\left(\q(\alpha P_1 +(1-\alpha) P_2)\right)\geq \alpha  \U\left(\q( P_1)\right) + (1-\alpha) \U\left(\q( P_2)\right),
\end{equation}
for all $\alpha\in [0,1]$.
\item \mg We assume that uncertainty cannot decrease for mixture of measurements. Therefore, $\U$  is concave with respect to the convex structure of measurements, i.e.,
\beq \label{concavityUmeasurement}
&\U\left(\q_X(P),\q_Y(P)\right)& \geq  \alpha  \U\left(\q_{X_1}( P),\q_Y(P)\right) \nonumber \\
&& + (1-\alpha) \U\left(\q_{X_2}( P),\q_Y(P)\right)
\eeq
where the observable $X$ is realized by the convex mixture of two observables $X_1,X_2$ with probability distribution $(\alpha,1-\alpha)$. \blk
\end{enumerate}

Now \emph{uncertainty of the statistics set}, $\U(\setth)$, is defined by the minimum $\U$ over all tuples of distributions in $\setth$,
\begin{equation}\label{eq:uncertaintyOFtheory}
\U(\setth)\eqdef \min_{x\in\setth} \U(x).
\end{equation}
\begin{itemize}
\item From concavity of $\U(\q(P))$ Eq.\eqref{concavityU} it follows that the minimum in Eq. \eqref{eq:uncertaintyOFtheory} is attained for the extremal points of $\setth$.  
\item Uncertainty measure possesses well-defined behavior under inlusion i.e.
i.e. for $S'\subset S$ we have 
\begin{equation}
\label{eq:precom-inclusion}
\U(S)\leq \U(S')
\end{equation}
\item It follows form postulate  Eq. \eqref{eq:indepStoch} that, any uncertainty measure is invariant under all doubly stochastic operations whose inverse is also a doubly stochastic operation. For instance, uncertainty is invariant under all possible relabeling (or permutations) of the outcomes. 
\end{itemize}  

\subsection{Strong Uncertainty}
We postulate any measure of strong (or full) uncertainty, which is denoted by $\U^f(\q(P))$, to be 
non-zero only if uncertainty is non-zero for all possible coarse-graining of outcome except the trivial one. Formally,  $\U^f(\q(P))=0$ if there exists extremal maps $\Lambda^e_1, \Lambda^e_2$ such that $\U((\Lambda^e_1,\Lambda^e_2)\q(P)) = 0$, where $\Lambda^e_1,\Lambda^e_2$ corresponds to the all possible permutations and coarse-graining except the trivial one. \\
Apart from that it also satisfies the postulates \eqref{eq:indepStoch}, \eqref{eq:coarse},  \eqref{concavityU} and \eqref{concavityUmeasurement} of uncertainty.

\subsection{Information Exclusion}
Here we list the postulates for any measure of Information exclusion of $\setth$.

\begin{enumerate}
\item $\E(\setth)\geq 0$ and $\E(\setth)=0$ if and only if for all outcome $k$, there exists a preparation, say $P_k$ such that 
\be
\q_X(k|P_k) = \tilde{\q}_Y(k|P_k) = 1,
\ee where $\tilde{\q}_Y(k|P_k)$ is an arbitrary $n$ element permutation of $\q_Y(k|P_k)$, i.e., $\tilde{\q}_Y(k|P_k) =  \q_Y(\pi(k)|P_k)$.
\item $\E(\setth)$ cannot decrease under doubly stochastic operations performed \emph{independently} on outcomes of observables  i.e. 
\begin{equation}
\E\left(( D_1,D_2) \setth \right) \geq \E(\setth) ,
\end{equation}
for all doubly-stochastic $n\times n$ matrices $D_1$ and $D_2$. Here, $\E\left(( D_1,D_2) \setth \right)$ denotes the allowed probability distribution in $S$ obtained from the observed statistics $( D_1,D_2)\q(P)$.
\item $\E(\setth)$ measure cannot increase under coarse-graining of outcomes. Formally,
\be
\E\left((\Lambda^e_1,\Lambda^e_2) \setth \right) \leq \E(\setth) \ee
for any \textit{extremal} stochastic maps $\Lambda^e_{1,2}$. 
\item $\E(\setth)$ measure possesses well-defined behavior under inclusion 
i.e. for $S'\subset S$ we have 
\begin{equation}
\label{e4}
\E(S)\leq \E(S').
\end{equation}
\item \mg Exclusion cannot decrease under convex mixture of measurements i.e.,
\beq \label{e5}
\E\left(S_{X,Y}\right)& \geq  \alpha  \E\left(S_{X_1,Y}\right) + (1-\alpha) \E\left(S_{X_2,Y}\right)
\eeq
where the observable $X$ is realized by the convex mixture of two observables $X_1,X_2$ with probability distribution $(\alpha,1-\alpha)$. \blk

\end{enumerate}

\section{Postulates for measures  of \precom\ and \com }
\label{sec:post-compl}
In this section we give the postulates that measures of \precom\ and \com\ for two observables. 

\subsection{\Precom} 



Recall that according to notation introduced in Section \ref{sec:framework},  $X \rightarrow Y$ means that for observables $X,Y$ 
there exists a stochastic map $\Lambda$ such that  $\q_{Y}(P) = \Lambda \q_{X}(P)$, \emph{simultaneously}, for all preparations $P$. \\

Now we propose that any measure of \precom\ ($\prec$) should depend only on the statistics that can be possibly observed   while measuring $X$ or $Y$, that is on the set $\setth$. Here are our postulates for the measure of \precom
\begin{enumerate}
\item We assume $\prec(\setth)\geq 0$ and that $\prec(\setth)=0$ if 
$X \rightarrow Y$ or $Y \rightarrow X$.
\item Any \precom\ measure is invariant under independent relabeling of outcomes of $X$ and $Y$ that is 
\begin{equation} \label{precom2}
\prec\left((\pi_1,\pi_2)\setth \right) = \prec(\setth)\ .
\end{equation}
for all permutations $\pi_{1,2}$ of $n$-element set. $(\pi_1,\pi_2)\setth$ denotes the allowed region obtained form the observed statistics of $(\pi_1,\pi_2)\q(P)$.
\item \Precom\ is a "monotonic" function of $\setgen$ under inclusion 
 i.e. for $S'\subset S$ we have 
\begin{equation}
\label{precom3}
\prec(S)\geq \prec(S')
\end{equation}

{\it Remark:} It might  seem natural to require monotonicity under post-processing, i.e. any stochastic map applied to outcomes of 
observables. However, it may happen that before processing  observables are in relation "$\to$", i.e. one can simulate the other one,
yet after some channel, they are not any more. Now, we require that \precom\ is zero for observables that are in relation, 
and the action of the channel can make it nonzero. Thus  \precom\ is not monotonic under post-processing. \mg Similarly, it might also seem that \precom\ cannot increase for convex mixture of two observables. However, one can find three observables such that $\prec(S_{X_1,Y})= \prec(S_{X_2,Y})=0$ but $\prec(S_{X,Y})>0$ where the observable $X$ is realized by convex mixture of $X_1,X_2$ (see Appendix \ref{appendix:non-extremal}). 
Yet for complementarity (see Section \ref{sec:post-compl2}) there is no such problem, 
and we will postulate its monotonicity under post-processing and non-increasing under convex mixtures.  \blk


\end{enumerate}


We now outline the postulates for other two measures of \precom.
Let $\prec^f$ be the measure of Full \precom, and it it non-zero only if for all possible nontrivial marginals of $\q_X(P),\q_Y(P)$ the \precom\ is non-zero. Formally, we require $\prec^f(\setth) = 0$, if there exists \textit{extremal} stochastic maps $\Lambda^e_1, \Lambda^e_2$ such that $\Lambda^e_1X \rightarrow \Lambda^e_2Y$ or $\Lambda^e_2Y \rightarrow \Lambda^e_1X$ (equivalently, $\prec((\Lambda^e_1,\Lambda^e_2)\setth) = 0$), where $\Lambda^e_1,\Lambda^e_2$ corresponds to the all possible permutations and coarse-graining except the trivial one. 
Apart from this, $\prec^f$ is required to fulfill postulates \eqref{precom2}- \eqref{precom3} of \precom.


Let's denote the measure of single-outcome \precom\ by $\prec^1(\setth)$. We require $\prec^1(\setth)=0$ if  there exists extremal stochastic maps $\Lambda^e_1, \Lambda^e_2$ that belongs to a class of coarse-grainings resulting binary outcome observable, in which exactly $(n-1)$ outcomes are coarse-grained to one outcome, such that $\prec((\Lambda^e_1,\Lambda^e_2)\setth) = 0$. In addition, it should also satisfy the other postulates \eqref{precom2}-\eqref{precom3} of \precom.

\subsection{Complementarity}
\label{sec:post-compl2}
The postulates for \com\ are as follows,
\begin{enumerate}\label{com1}
\item $\C_{X,Y} \geq 0$, and $\C_{X,Y}=0$ if there exists another observable $Z$ in the theory such that $Z\rightarrow X$ and $Z\rightarrow Y$. 
\item An measure of \com\ cannot increase if instead of $X$ and $Y$ we have only access to statistics of post-processed observables. Mathematically, this corresponds to 
\begin{equation} \label{com2}
\C_{X,Y} \geq \C_{\Lambda_1X, \Lambda_2Y} 
\end{equation}
where $\Lambda_{1,2}$ are arbitrary stochastic $n\times n$ matrices.
As a consequence, any \com\ measure is invariant under stochastic maps $\Lambda_{1,2}$ whose inverses are also stochastic maps. For example, $
\C_{\pi_1X,\pi_2Y} = \C_{X,Y} $
for all permutations $\pi_{1,2}$ of $n$-element set. 
\item \mg Complementarity cannot increases under mixture of observables, i.e.,
\beq \label{com3}
\C_{X,Y} \leq  \alpha  \C_{X_1,Y}+ (1-\alpha) \C_{X_2,Y}
\eeq
where the observable $X$ is realized by the convex mixture of two observables $X_1,X_2$ with probability distribution $(\alpha,1-\alpha)$. \blk
\end{enumerate}

\begin{rem*}
Qualitatively, postulate 3 can be justified by postulate 1. Specifically if observables $X_{1,2}$ are not complementary with $Y$ (i.e. $\C_{X_1,Y}=\C_{X_2,Y}=0$), then observable $X$, realized by their a convex mixture (with weights $\alpha$ and $1-\alpha$ respectively), is also not complementary with $Y$. Indeed as a mother obsevable of $X$ and $Y$ one can take a mixture (with the same weights as above) of mother observables $O_1$ and $O_2$ of pairs $(X_1, Y)$ and $(X_2,Y)$. This works because without loss of generality the stochastic  maps $O_1\to (X_1, Y)$, $O_1\to (X_2, Y)$ can be take as simply taking marginals.

\end{rem*}


\mg
To see the connection between  \precom\ and \com\, recall  that the former can be used to define the latter. Concretely,  using the prescription  from Eq.\eqref{eq:compFROMind}  we obtain that any measure of \precom\ defines  we need the following notions.
Now given a \precom\ measure one can obtain \com\ measure for two extremal observables as follows,
\be\label{com-ind-ext}
\C_{X,Y} = \min_{X':X'\rightarrow X} \min_{Y':Y'\rightarrow Y} \prec(S_{X',Y'})\ ,
\ee
where the infimum is taken over all pairs of observables $X',Y'$ that simulate a pair $X,Y$. For general observables, we follow the convex-roof extension of the above definition \eqref{com-ind-ext}. Formally, 
\be \label{com-non-ext}
\C_{X,Y} = \min_{\{\alpha_i,X_i\}} \min_{\{\beta_j,Y_j\}} \sum_{i,j} \alpha_i \beta_j\ \C_{X_i,Y_j}\ ,
\ee
where the minimum is taken over all possible decomposition of the observables $X,Y$  to the extremal observables, i.e.,
\be
\forall P\in \P, \ \q_X(P) = \sum_i \alpha_i \q_{X_i}(P), \ \q_Y(P) = \sum_j \beta_j \q_{Y_j}(P).
\ee
We can express \com\ from a measure of \precom\ in the explicit form as follows, 
\be  \label{com-ind-gen}
\C_{X,Y} = \min_{\{\alpha_i,X_i\}} \min_{\{\beta_j,Y_j\}} \sum_{i,j} \alpha_i \beta_j \min_{X_i':X_i'\rightarrow X_i} \min_{Y_j':Y_j'\rightarrow Y_j}   \prec(S_{X_i',Y_j'})\ .
\ee
 Thus, for clean and extremal observables, $\C=\prec$.
Note, that while \precom\ was not required to be monotonic under stochastic maps, due to the above definition of \com\ 
it will be natural require such monotonicity.
 
\blk

{\it Remark.} As we have said in Sec. \ref{sec:basic-indep},
the simplest theory, for which \com\ is not equal to \precom\ is the already mentioned two bits with three observables: 
$X$ for the  first bit, $Y$ for the second bit, and third observable $Z$ with four outcomes, that measures value of both bits. 
The two observables $X$ and $Y$  are clearly independent for any possible measure, while both they come from $C$ by post-processing,
so that they are not clean, and \com\ vanishes.


Similarly, one can set the postulates of the measures of Full \com\ and single-outcome \com. We denote the measures by $\C^f_{X,Y}$, and $\C^1_{X,Y}$ respectively.  $\C^f(\setth) = 0$, if there exists \textit{extremal} stochastic maps $\Lambda^e_1, \Lambda^e_2$ such that $\C((\Lambda^e_1,\Lambda^e_2)\setth) = 0$), where $\Lambda^e_1,\Lambda^e_2$ corresponds to the all possible permutations and coarse-graining except the trivial one. While
 $\C^1_{X,Y}=0$ if there exists extremal stochastic maps $\Lambda^e_1, \Lambda^e_2$ that belongs to a class of coarse-grainings resulting binary outcome observable, in which exactly $(n-1)$ outcomes are coarse-grained to one outcome, such that $\C((\Lambda^e_1,\Lambda^e_2)\setth) = 0$.
Additionally, both the measures should satisfy the postulates of non-increasing under post-processing \eqref{com2}.

Being the notion of \com\ is associated with the notion of joint measurability, the foremost measure of it that comes to our mind is the robustness parameter with respect to the white noise. This measure has been generalized in the context of general operational theory in \cite{Busch13}. Given two extremal observables $X,Y$, we define another two observables $X^{\lambda},Y^{\lambda}$ such that
\beq
\q_{X^{\lambda}} (P) = (1-\lambda) \q_X(P) + \frac{\lambda}{d} (1,...,1),\nonumber \\
 \q_{Y^{\lambda}} (P) = (1-\lambda) \q_Y(P) + \frac{\lambda}{d} (1,...,1)
\eeq
taking $\lambda\in [0,1]$ be the parameter of white noise.
The measure of \com\ is defined to be the minimum value of $\lambda$ for which there exists another observable $Z$ in the theory such that $Z\rightarrow X^{\lambda},Y^{\lambda}$, i.e., $\C_{X^{\lambda},Y^{\lambda}}=0$. \mg For non-extremal observables we consider the convex-roof extension \eqref{com-non-ext} \blk. It can be readily verified that this measure satisfies the other postulates of \com.
  The first postulate follows from its definition. Further, suppose the complementarity of two observables $X,Y$ is $\lambda_C$, i.e., $\C_{X^{\lambda_C},Y^{\lambda_C}}=0$, then we know that $\C_{\Lambda_1X^{\lambda_C},\Lambda_2Y^{\lambda_C}}$ is also 0. Thus, $\C_{\Lambda_1X,\Lambda_2Y}$ cannot be larger than $\lambda_C$.


\section{Measures of uncertainty and \precom}
\label{sec:measures-uncert}
In this section, we propose some measures of uncertainty and \precom.



\subsection{\Com\ and Uncertainty measures based on random access code}
\label{sec:random-access}
We propose a measure of \precom\ based on a communication tasks known as random access code \cite{Ambainis2008}. This task involves two devices, preparation and measurement, possessed by Alice and Bob respectively. In each round of the task, Alice receives a two dit input $a=(a_1a_2) \in \{1,...,d\}^2$, prepares a $d$-dimensional system, say $P_a$, and sends to Bob. Bob receives the communicated system from Alice and measures an observable depending on his obtained input $b\in \{1,2\}$. He wants to guess $a_b$. Let us denote the probability of giving the correct answer for input $a,b$ is $p(a_b|a,b)$. 
A figure of merit of such communication task can be any reasonable function of these probabilities, $\mathcal{F}\{p(a_b|a,b)\}$. For instance, it could be the average success probability of guessing $a_b$,
\be \label{avgrac}
p_s = \frac{1}{2d^2}\sum_{a,b} p(a_b|a,b)
\ee 
where the inputs are uniformly distributed.  

In most common version of the above task, Bob is free to choose the optimal observables that would maximize the probability of success. 
Here, to connect the task with complementarity,  we will fix the Bob's observables to be one of two observables $X$ and $Y$.
For convenience let us denote $X=X_1$ and $Y=X_2$. Now, 
for input $a,b,$ Bob obtains a statistics $\q_{X_b}(P_a)$ where $X_b$ denotes the $d$-outcome observable measured on $P_a$. He can apply some post processing after the measurement, and thus the obtained probability for correct answer is,
\be
p(a_b|a,b) = \tilde{\q}_{X_b}(a_b|P_a), \text{ where } \tilde{\q}_{X_b}(P_a) = \Lambda_b \q_{X_b}(P_a). 
\ee
Now given any theory and the two observables $X_1,X_2$, the relevant quantity $p_s=\mathcal{F}\{p(a_b|a,b)\}$ is maximized over all possible
$P_a,\Lambda_b$. Let the measure of \precom\ of these two observables be as follows,
\be \label{racm}
\prec(S_{X_1,X_2}) = \frac{p_s(X_1,X_2) - \max\big( p_s(X_1), p_s(X_2)\big)}{1-\max\big( p_s(X_1), p_s(X_2)\big)}, 
\ee
where $p_s(X_1,X_2)$ denotes the optimal value of the figure of merit when Bob has access to two observables $X_1,X_2$ and $p_s(X_1)$ denotes the same when Bob has access to only $X_1$. Note that, $\prec(S_{X_1,X_2})$ is normalized, i.e., it takes value within the range $[0,1]$.

One can readily check that the measure \eqref{racm} satisfies the postulates of \precom.  Since Bob is allowed to apply arbitrary stochastic may $\Lambda_b$, $p_s(X_1,X_2)$ is eventually equal to $p_s(X_1)$ (or $p_s(X_2))$ if $X_1\rightarrow X_2$ (or $X_2\rightarrow X_1)$. Due to 
the same reason, it is invariant under permutation. Further, as $\mathcal{F}\{p(a_b|a,b)\}$ is maximized over all possible preparations $P_a$, 
 it is monotonic under inclusion. 

This measure relates \precom\  to efficacy of an operational task. However, it is not a measure of full \precom. In future, one may look for similar operational task that quantifies full \precom. 

One can define a measure of uncertainty based on the same communication task. In this situation, Bob is allowed to apply only doubly stochastic map on the observed statistics after measurement. The measure of uncertainty for a preparation $P$ is considered to be converse of the maximum success probability of guessing $a_b$ over all possible inputs $a$, 
\beq  \label{u:rac}
&&\U(\q_{X_1}(P),\q_{X_2}(P)) =1-\max_{a} \frac{1}{2}\sum_b p(a_b|a,b) \nonumber \\
&& =1-\max_{a} \frac{1}{2} (q_{X_1}(a_1|P) + q_{X_2}(a_2|P)).\eeq
Subsequently, following \eqref{eq:uncertaintyOFtheory}, the uncertainty of the statistics set
\be \label{urac}
\U(S_{X_1,X_2}) = 1- \max\limits_{P_a\in \mathcal{P}} \frac{1}{2} \sum_b q_{X_b}(a_b|P_a).\ee 
By the definition the above measure \eqref{u:rac} is zero if and only if the distribution of $\q(P)$ is deterministic and cannot decrease under doubly stochastic map. The measure of uncertainty can be rewritten as, $\min\limits_{P_a\in \mathcal{P}}(1-  \frac{1}{2} \sum_b q_{X_b}(a_b|P_a)).$ Since $1-\frac{1}{2} \sum_b q_{X_b}(a_b|P_a)$ is linear with respect to a convex mixtures of two preparations and the minimum function of two linear functions is concave, it satisfies \eqref{concavityU}. It can also be readily checked that  $\frac{1}{2} \sum_b q_{X_b}(a_b|P_a)$ cannot decreases under coarse-graining of observables, and  therefore it satisfies monotonicity under coarse-graining \eqref{eq:coarse}. \mg To see that the measure also satisfies convexity  \eqref{concavityUmeasurement}, we express the uncertanity measure \eqref{u:rac} between $X_2$ and a convex mixture of two observables $X_1,X'_1$ with probability distribution $(\alpha,1-\alpha)$ in the following way,
\beq 
&& 1 - \max_{a} \frac{1}{2} (\alpha q_{X_1}(a_1|P) + (1-\alpha) q_{X'_1}(a_1|P)  + q_{X_2}(a_2|P)) \nonumber \\
& \geq&  \alpha (1 - \max_{a} \frac{1}{2} (q_{X_1}(a_1|P) + q_{X_2}(a_2|P)) \nonumber \\
&& + (1-\alpha) (1 - \max_{a} \frac{1}{2} (q_{X'_1}(a_1|P) + q_{X_2}(a_2|P)) \nonumber \\
& =& \alpha \U(\q_{X_1}(P),\q_{X_2}(P)) + (1-\alpha)\U(\q_{X'_1}(P),\q_{X_2}(P)).
\eeq
\blk

Let us remark, that a variant of the obtained measure of uncertainty was considered e.g. in \cite{OW}. Here we have pointed out
its operational origin (by connecting it to random access code), as well as shown that it satisfies the postulates. 
Similarly, we define the measure of information exclusion as the converse of the average success probability of guessing $a_b$ restricted to those inputs when $a_1=a_2$,
\beq \label{erac}
&\E(S_{X_1,X_2}) &= 1 - \frac{1}{2d} \sum_{b,a|a_1=a_2} p(a_b|P_a) \nonumber \\
&&= 1 - \frac{1}{2d} \max_{P_a\in \mathcal{P}} \sum_{b,a|a_1=a_2} \tilde{q}_{X_b}(a_b|P_a).
\eeq 
taking into account $\tilde{\q}_{X_b}(P_a) = \pi \q_{X_b}(P_a)$. \\

{\it Example: quantum theory.} 
To provide a complete example in quantum theory, we take the figure of merit as the average success probability \eqref{avgrac}.
 It has been shown that the optimal value for classical system \cite{classicalrac}
\be \label{racclassical}
p_s(X_1)=\frac{1}{2}+\frac{1}{2d}.
\ee 
For the two quantum projective measurements correspond to the basis $X_1 = \{|i\rangle\}^{d}_{i=1}$ and $X_2 = \{|\psi\rangle_j\}^{d}_{j=1}$ accessed by Bob, the average success probability \eqref{avgrac},
\be \label{avgq}
p_s(X_1,X_2) = \frac{1}{2} + \frac{1}{2d^2} \sum_{a_1,a_2} |\<a_1|\psi_{a_2}\>| .
\ee 
The proof of this fact is given in the appendix \ref{appendix:qrac}. The left-hand-side of \eqref{avgq} is strictly great than $p_s(X_1)$ \eqref{racclassical} for any two distinct quantum observables since
\be
\sum_{a_1,a_2} |\<a_1|\psi_{a_2}\>| > \sum_{a_1,a_2} |\<a_1|\psi_{a_2}\>|^2 =d.
\ee
The optimal quantum value of $p_s=\frac{1}{2}+\frac{1}{2\sqrt{d}}$ that corresponds to two mutually unbiased basis \cite{rac2}.

Hence, the \precom\  measure \eqref{racm} based on random access code for two $d$-dimensional quantum observables is given by,
\be
\label{eq:ind-rac}
\prec(S_{X_1,X_2})  = \frac{1}{d-1} \left(\frac{1}{d}\sum_{a_1,a_2} |\<a_1|\psi_{a_2}\>| -1 \right).
\ee
Further, invoking \eqref{eq:qrac1} one obtains the uncertainty measure \eqref{urac},
\be
\U(S_{X_1,X_2}) = \frac{1}{2} (1- \max_{a_1,a_2} |\<a_1|\psi_{a_2}\>|),
\ee
and the information exclusion measure \eqref{erac},
\be
\label{eq:E-rac}
\E(S_{X_1,X_2}) = \frac{1}{2}\left(1 - \frac{1}{d} \max_{\pi} \sum_{i=\pi(j)} |\<a_i|\psi_{a_{\pi(j)}}\>| \right)
\ee
where $\pi$ is $d$-element permutation.

\subsection{Re-scaling and volume of the probability space} 
\label{sec:rescaling}
We shall now define measure of independence, by means of rescalings of the statistics set $\sxy$.
\begin{definition} 
\Precom\ is given by maximal $r \in [0,1]$ such that  
\mg $r \setgen + x \subset \setth $. \blk
I.e. $r$ is maximum rescaling factor of the full set $\setgen$  
such that the rescaled set \mg $r \setgen + x$ is contained in
 $\setgen$ after shifting along some vector $x$.\blk We denote it by $\prec_r$.
\end{definition}

It is clear from the definition that $\prec_r$ is invariant under permutation \eqref{precom2} and monotonic under inclusion \eqref{precom3}. We know that the dimension of $S$ is $2(d-1)$. If $\Lambda^e_1X\rightarrow \Lambda^e_2Y$ for some extremal stochastic maps $\Lambda^e_{1,2}$, then the number of independent variables to specify $\q(P)$ is less than $2(d-1)$. It follows that the dimension of the statistics set  $\sxy$ is strictly less than $2(d-1)$, 
thereby $\prec_r=0$. Thus, $\prec_r$ is a good measure of full \precom. 
For instance, $S$ being a square (i.e., two binary observables) the full set s-bit and classical c-bit have complementarities $1, 0$ respectively. In the case of quantum, consider qubit observables $Z=\sigma_z$, $X={\bf n}\cdot \sigma$ with $n_y=0$, and $n_x^2 + n_z^2=1$. 
The boundary of the statistics set of possible pairs of averages $(\<\psi|Z|\psi\>,\<\psi|X|\psi\>)$ is given by,
\be \label{quantumbinary}
\frac{(x+z)^2 }{2 a^2} + \frac{(x-z)^2 }{2 b^2} = 1,  \ \text{with} \ a= \frac{n_x}{\sqrt{1-n_z}},   b= \frac{n_x}{\sqrt{1+n_z}}.
\ee  
It is shown in Fig. \ref{ellipse}. The parameters $a$ and $b$ are the major and minor semiaxes of the ellipse, respectively. Thus, the diagonal of the largest square inside the body is $2b$. Subsequently, a simple calculation leads to 
\be  \label{kCbinary}
\prec_r=\frac{\sqrt{2}b}{2}=\frac{n_x}{\sqrt{2(1+n_z)}}.
\ee Note that with this definition, q-bit does not have maximal possible complementarity as s-bit.
\begin{figure}[http]
\begin{center}
\includegraphics[scale=0.41]{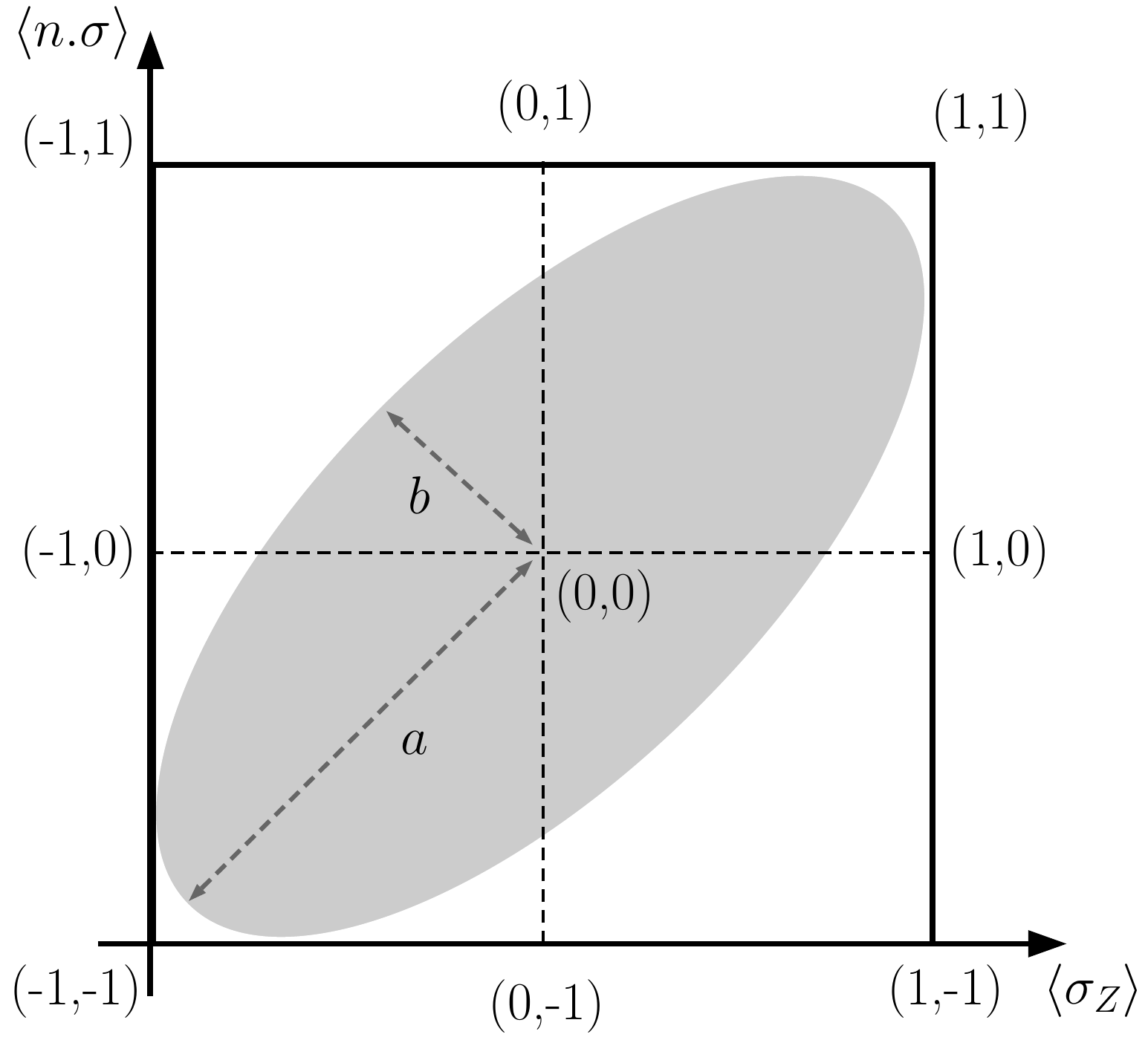}
\end{center}
\caption{ The statistics set of two quantum observables $\sigma_z$ and $n_x \sigma_x+n_z\sigma_z$ is presented in gray. The semi-major and semi-minor axes are denoted by $a,b$ respectively.}
\label{ellipse}
\end{figure}

Following the same arguments, one can see that the volume of $\setth$ is also a measure of full \precom. For s-bit, q-bit observables $(Z,X)$, and c-bit the volume of $\setth$ are 4, $\pi ab = \pi n_x$, 0 respectively.

\begin{figure}[http]
\begin{center}
\includegraphics[scale=0.7]{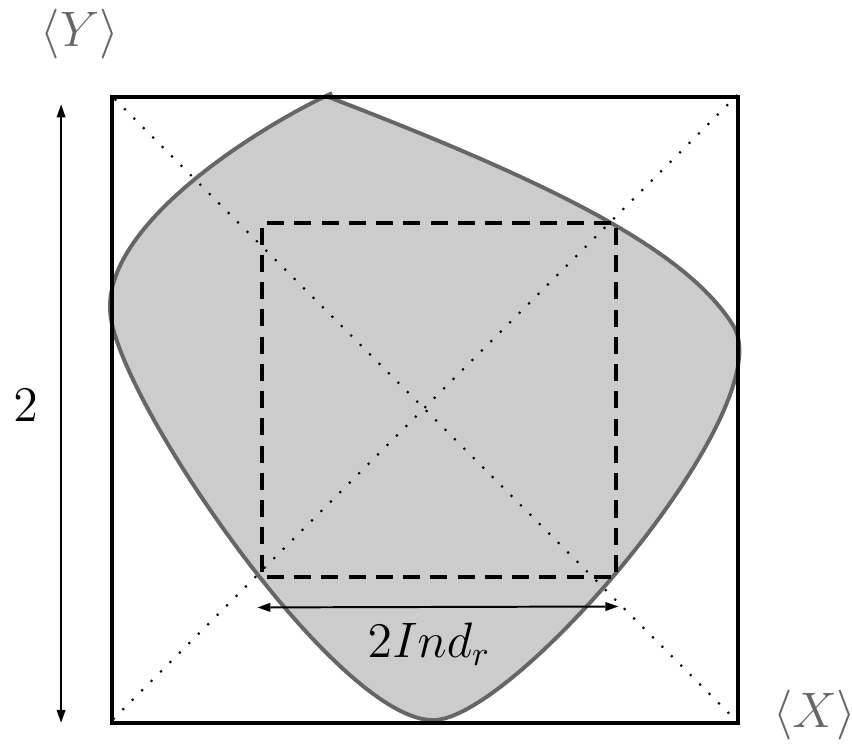}
\end{center}
\caption{The re-scaling measure of \precom\ ($\prec_r$) of the statistics 
set for binary outcome observables $X,Y$. The statistics set is presented in gray.}
\label{rescaling}
\end{figure}

Here, it can be noted that the uncertainty ($\U$) \eqref{urac} proposed in \ref{sec:random-access} is zero for s-bit and c-bit. While for the quantum observables in Fig.\ref{ellipse}, 
\beq \label{kUbinary}
&\U &= 1- \max\limits_{P_a\in \mathcal{P}} \frac{1}{2} \sum_b q_{X_b}(a_b|P_a) \nonumber \\
&& = 1-\frac{a}{\sqrt{2}}= 1-\frac{n_x}{\sqrt{2(1-n_z)}}.\eeq


\subsection{Complementarity measures based on preimage}
\label{sec:preimage}
In this section, we shall propose just a scheme of building various measures  of \precom\ from 
a class of functions defined on joint distributions. Namely, we will require from 
such a function  
that it vanishes on distributions of the form $p(i,j)=p(i,i)\delta_{ij}$. 
We shall slightly abuse notation, by naming such functions also "independence" (now not independence of a pair of observables, 
but independence of joint distribution). 
An example of independence measure is the so called {\it variation of information}:
\be
VI(p_{XY})=H(X|Y)+H(Y|X)
\ee where $H(\cdot)$ is the entropy. 
To define \precom\ on pairs of observables from that defined on joint distirbution we proceed as follows.
Fix some set $S_{pre}$ to be a convex set of joint distributions, whose marginals give rise to  $\setth$. 
Let us fix two channels $\Lambda_1$ and $\Lambda_2$ acting on the outputs of observables $X$ and $Y$ respectively. 
We now consider a set $S_{pre}(X,Y,\Lambda_1,\Lambda_2)$ (in short $S_{pre}$) 
of joint distributions which after applying local processing $\Lambda_1\otimes \Lambda_2$,
where $\Lambda_i$ are channels, gives rise to $\setth$ via marginals. In another words, each element of $\setth$ 
is a pair of marginals of some distribution from $S_{pre}$ subjected to $\Lambda_1\otimes \Lambda_2$, and vice versa, 
if we apply $\Lambda_1\otimes \Lambda_2$ channel  to each joint distribution from $S_{pre}$, the pair of marginals of the obtained 
distribution belongs to $\setth$.

The \precom\ measure is now defined as
\be
\prec(X,Y)= \min_{S_{pre}}\max_{p \in S_{pre}} Ind(p)
\ee
where the minimum is taken over all convex sets $S_{pre}$ of distributions, such that there exist channels $\Lambda_1$ and $\Lambda_2$ 
for which $S_{pre}$ that give rise to $\setth$, as described above. 

Let us see that the measure satisfies the postulates for \precom. Suppose that one observable is a processed version of the other,
i.e. can be obtained from the other via some channel $\Lambda$.
Then we can take the preimage to be the set of perfectly correlated distributions, with the choice $\Lambda_1\ot\Lambda_2=I\ot\Lambda$.
Hence all the distributions from preimage have vanishing independence, so that the measure vanishes. 
By definition, if we enlarge the set $\setth$, the measure can only increase, as the preimage cannot decrease. 
Thus we obtain that the second postulate is satisfied too.

We illustrate the concept of the above measure by means of two examples:
the classical bit Fig. \ref{fig:compl}a) and "diamond" bit, in Fig. \ref{fig:compl}b),
where we take the variation of information as independence measure of joint distributions. 

For the classical bit (two identical observables) the set $\setth$ can be  obtained as an image 
of an edge of the tetrahedron, which allows only for perfectly correlated distributions, hence the measure vanishes.

Let us argue, that the set depicted in Fig.  \ref{fig:compl}b) is the only possible preimage.
Note, first that corners of the diamond are the following pairs of distributions (we use quantum notation just for brevity
\be
(I/2, |0\>\<0|), (I/2, |1\>\<1|), (|0\>\<0|,I/2),  (|1\>\<1|,I/2).
\ee

Since always one of the distribution in the pair is pure, the only joint distributions that return these pairs via marginals are product. 
Let us argue, that for any fixed pair of channels $\Lambda_1\ot \Lambda_2$, the distributions that can give rise 
through these channels to product distributions must be product too.
To this end, note that if we start with correlated distribution, and act with product channel, 
the output distribution is product if and only if, at least one of the channels is "information killing", i.e. 
it produces a single state for all input states. 
Clearly none of our channels can be like that, because sometimes we need to produce $I/2$ and sometimes $|0\>\<0|$ or $|1\>\<1|$. 
Thus, the initial joint distributions must be product. 

The channel $\Lambda_1$ has just to send two of distributions to $I/2$, one to $|0\>\<0|$ and one to $|1\>\<1|$, (the same about channel $\Lambda_2$). Suppose that distribution sent to $|0\>\<0|$ is neither  $|0\>\<0|$ nor  $|1\>\<1|$. 
Then one directly checks that that channel send all the states to $|0\>\<0|$, which cannot be so (as we want also to get  $|1\>\<1|$
and $I/2$ for some input states. 
Thus the input must be either $|0\>\<0|$ or  $|1\>\<1|$. Suppose it is $|0\>\<0|$.
Then one finds that the channels is of the form
\be
\left[\bea{cc}
1 & q \\
0 & 1-q \\
\eea \right]
\ee
Now this channel must produce $|1\>\<1|$  out of some state. One finds then, that the channel must be identity. 
If the input is $|1\>\<1|$  we obtain that the channel is flip. 
Similarly $\Lambda_2$ is either identity of a flip. 
Thus the preimage of the four corners of the diamond are the products 
\be
I/2 \otimes |0\>\<0|,\quad I/2\ot |1\>\<1|, \quad |0\>\<0|\otimes I/2,  |1\>\<1|\otimes I/2.
\ee
Hence the preimage, since it is a convex set by definition,  contains $I/2\ot I/2$ as an equal mixture of the above distributions.
We conclude that the measure of \precom\ is equal to $1$.

\begin{figure}[bth!]
\begin{center}
\includegraphics[scale=0.3]{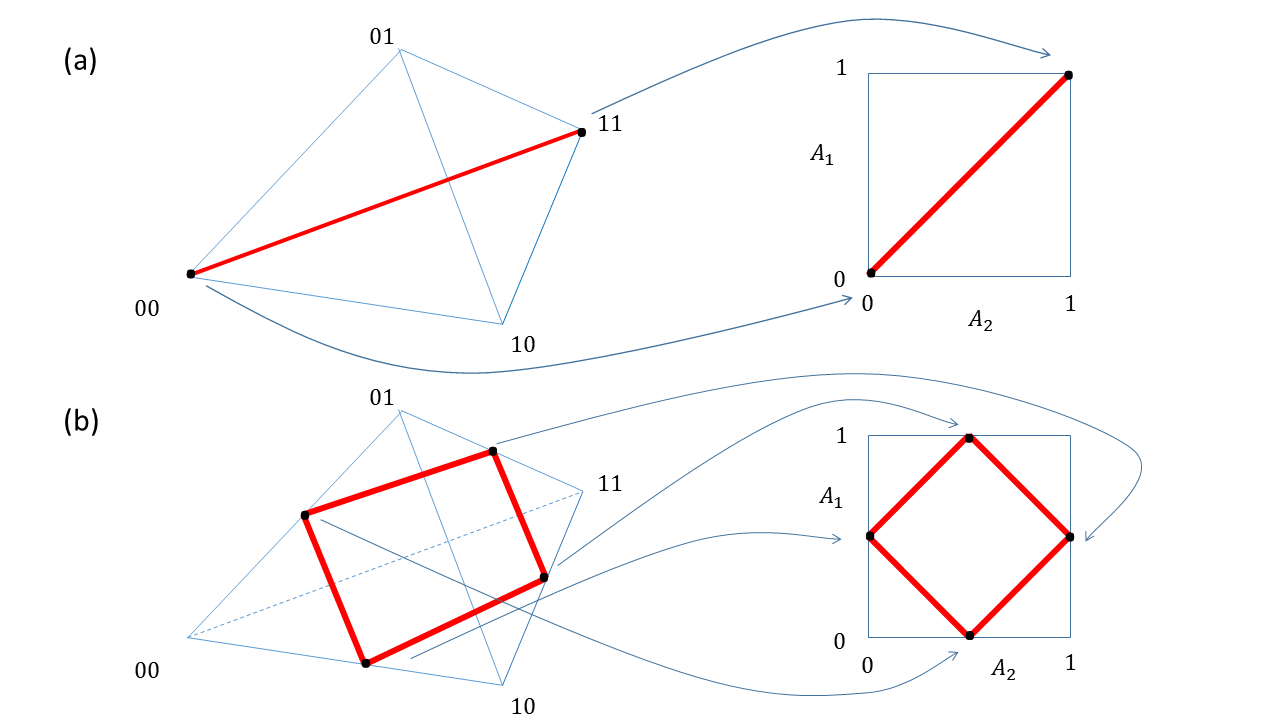}
\end{center}
\caption[]{\label{fig:compl} Examples of a classical bit and "diamond". (a) \Precom\ of two identical observables is zero, since it is obtained as an image of perfectly correlated  probability distributions. (b) \Precom\ of observables for which the statistics constitute diamond  have \precom\ equal to $1$. Preimage is a square that contains the center of the product of the product of simplices, which has 
\precom\ $1$. }
\end{figure}



\section{Preparation uncertainty relation}
\label{sec:pur}
As proposed in Section \ref{sec:compl-unc-general} from measures of uncertainty and complementarity, one can build uncertainty relations 
of the form 
\be
\label{pur}
\U(\setth)\geq f^{\uparrow}(\C(\setth)).
\ee
where $f^{\uparrow}$ is  non-decreasing functions whose range is non-negative.

We first note that such uncertainty principle is not  satisfied in  all theories. E.g. square bit, whose statistics set is the whole square 
cannot satisfy the above uncertainty relations for any measures of complementarity. Indeed from postulates 
it follows that  if $\sxy$ is the whole square, then there is no uncertainty of any kind, as it contains all corners. Also, complementarity,
by monotonicity under inclusion must be maximal possible.  Therefore, any complementarity measure (apart from trivial one that
is zero for all possible sets) will be nonzero.




\subsection{PUR from random access codes} 

We derive here PUR constructed out of measures of uncertainty and complementarity
in terms of random access codes from Section \ref{sec:random-access}. This PUR is actually Exclusion Principle of the similar 
form as that of  \cite{GrudkaExclusion}. 

\begin{fact}
In quantum mechanics the following PUR holds for arbitrary two observables $X$ and $Y$
with $d$ outcomes, with one dimensional eigenprojectors:
\begin{equation}\label{eq:exclusionURrac-fact}
E(S_{X,Y})\geq \frac{(\C_{X,Y})^2}{4d}\ . 
\end{equation}
where  $E$ is measure of exclusion of \eqref{eq:E-rac} and $Ind(\setth)=\C_{X,Y} $ is measure of independence of  \eqref{eq:ind-rac}. 
\end{fact}
Of course, since  the considered observables are clean \mg and extremal \blk $Ind$ is same as complementarity. 

\subsection{PUR from rescaling} 

Here we consider the re-scaling measures of \com\ ($\prec_r = \C_r$ for clean \mg and extremal \blk observables) 
and uncertainty $\U$ mentioned in Section \ref{sec:rescaling} to provide an example of PUR between binary observables.

\begin{fact}
Two quantum binary observables $Z=\sigma_z$, $X={\bf n}\cdot \sigma$ with $n_y=0$, and $n_x^2 + n_z^2=1$, satisfy the following PUR,
which is even in a form of equality:
\be
\label{eq:urq2}
 \C^2_r + (1-\U)^2 = 1.
\ee
\end{fact}
{\it Proof.} First one can express $\prec_r$ and $\U$ in \eqref{kCbinary}-\eqref{kUbinary} in terms of only $n_z$ 
by substituting $n_x=\sqrt{1-n^2_z}$. Further, by equalizing $n_z$ as a function of $\C_r$ and $\U$, one obtains the above PUR 
with equality.

\subsection{Reverse PUR from rescaling}
\label{subsec:reverse-rescaling}
In sec. \ref{subsec:reverse} we introduced the concept of reverse uncertainty relation. As said there,
unlike the uncertainty relation, which may or may not hold in a given theory, the reverse one is expected to 
hold almost by definition in any theory. 
Here we present such a relation in the case of binary outcomes, for the uncertainty based on rescaling:
\begin{fact}
In any theory, for any two binary sharp, clean and extremal observables, the following reverse PUR holds
\be
2\C_r\geq \U\ .
\ee
\end{fact}

The proof is given in Appendix \ref{app:reverseUR}.


\subsection{Uncertainty relation from physical principles} 
\label{sec:icp}

Now, we shall show how  the information theoretic principle namely {\it Information Contents Principle}  \cite{Czekaj-ICP} - 
a single system version of   {\it Information Causality}  \cite{ic} imposes PUR on the physical theories. 
Likewise, one can postulate PUR or obtain PUR from other principles which should be obeyed by any physical theories. 
\begin{figure}[http]
\begin{center}
\includegraphics[scale=0.43]{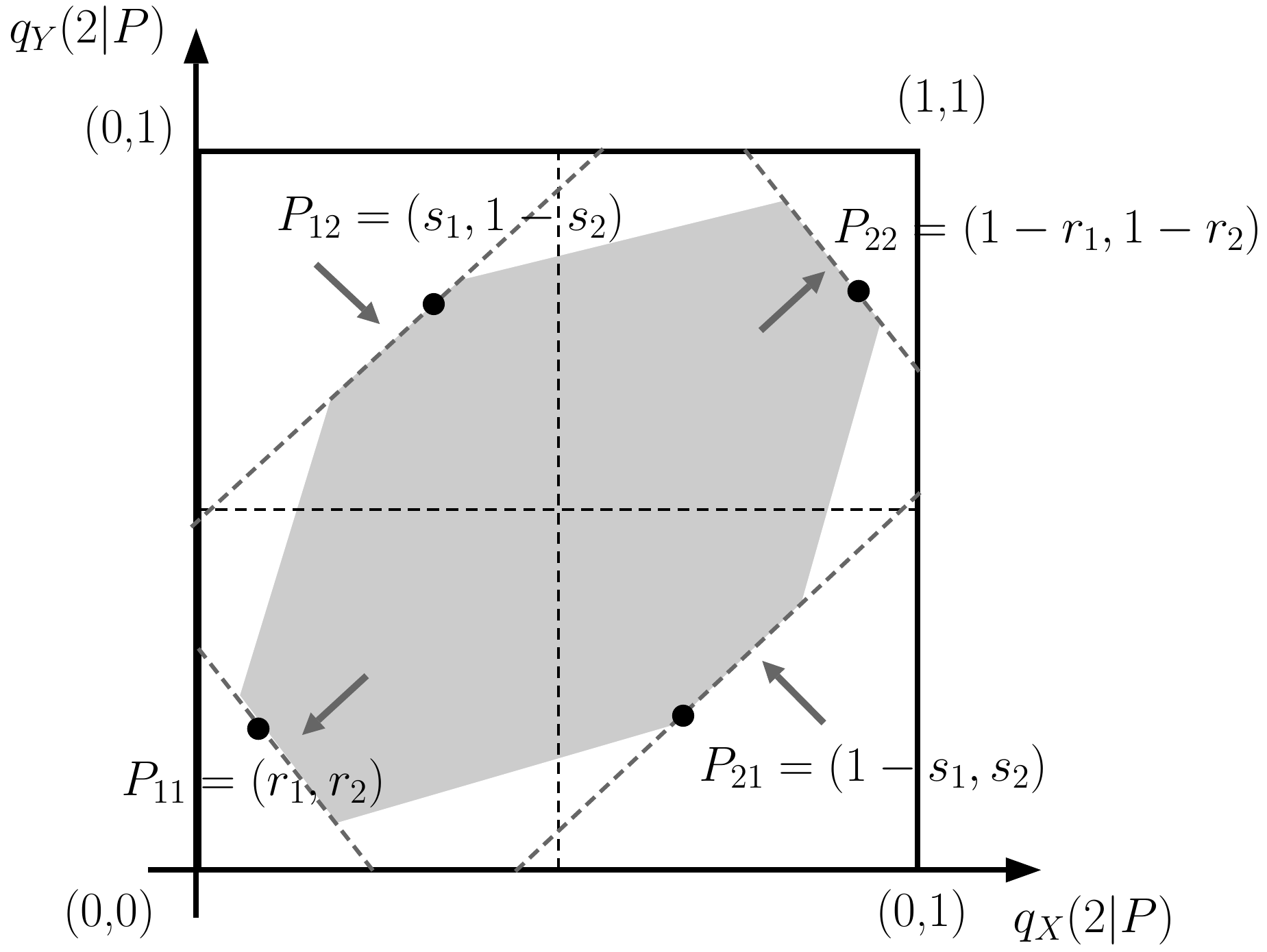}
\end{center}
\caption{\label{fig:icp} The statistics set $\setth$ possesses the symmetry under the reflection of the diagonal of the square. State of the system $P_{a_1a_2}$ is described by the pair of probabilities $(q_X(2|P),q_Y(2|P))$.}
\end{figure}

Let us recall the communication task random access code presented before. We assume the inputs $a,b$, given to Alice and Bob, are uniformly distributed and uncorrelated, i.e., $\forall a,b,\ p(a,b)=p(a)p(b), p(a)=1/4,p(b)=1/2$. We denote the classical output of Bob by $C_b$ for his input $b$.  The information causality provides a bound on the correlations as follows,
\beq \label{icp}
& I(C_1:X) + I(C_2:Y) - I(C_1:C_2) \leq 1 \\ \nonumber
& \implies H(X) - H(C_1X) + H(Y) - H(C_2Y) + H(C_1C_2) \leq 1 .
\eeq

Since we deal with two binary outcome measurements, the statistics set $\setth$ can be conveniently presented by the pair of probabilities $(q_X(2|P),q_Y(2|P))$ as shown in Fig. \ref{fig:icp}.
For the sake of simplicity, we consider a class of theories in which the statistics set $\setth$ possesses symmetry under permutation of outcome, i.e., for all $q_X(P)$ there exists another preparation $P'$ such that $q_Y(P') = q_X(P)$ and vice versa. In other words, $\setth$ is symmetric with respect to the diagonal of the square. 
 Due to the symmetry of the statistics set in Fig.\ref{fig:icp}, for a preparation with statistics $(q_X(2|P),q_Y(2|P))=(r_1,r_2)$, we know there another preparation with $(q_X(2|P'),q_Y(2|P'))=(1-r_1,1-r_2)$. Accordingly, we obtain the probability distribution for $C_1X$ and $C_2Y$,
\begin{center}
\begin{tabular}{c|c |c}
 & $X=1$ & $X=2$ \\ 
 \hline
$C_1=1$ & $\frac{1}{2}(\frac{r_1}{2}+\frac{s_1}{2})$ & $\frac{1}{2}(1-\frac{r_1}{2}-\frac{s_1}{2})$ \\
$C_1=2$ &  $\frac{1}{2}(1-\frac{r_1}{2}-\frac{s_1}{2})$ & $\frac{1}{2}(\frac{r_1}{2}+\frac{s_1}{2})$\\
\end{tabular}
\begin{tabular}{c|c |c}
 & $Y=1$ & $Y=2$ \\ 
 \hline
$C_2=1$ & $\frac{1}{2}(\frac{r_2}{2}+\frac{s_2}{2})$ & $\frac{1}{2}(1-\frac{r_2}{2}-\frac{s_2}{2})$ \\
$C_2=2$ & $\frac{1}{2}(1-\frac{r_2}{2}-\frac{s_2}{2})$ & $\frac{1}{2}(\frac{r_2}{2}+\frac{s_2}{2})$\\
\end{tabular}
\end{center} 
Thus, 
\be \begin{split}
&H(X)=H(Y)=\frac{1}{2}H(C_1C_2)=1,\\
&H(C_1X) = h\left(\frac{r_1}{2}+\frac{s_1}{2}\right)+1,H(C_2Y) = h\left(\frac{r_2}{2}+\frac{s_2}{2}\right)+1,
\end{split}\ee where $h(p)=-p\log(p) - (1-p)\log(1-p)$. Substituting these expressions in the ICP \eqref{icp} we obtain the following relation,
\be
\label{eq:icp_to_ur}
h\left(\frac{r_1}{2}+\frac{s_1}{2}\right)+ h\left(\frac{r_2}{2}+\frac{s_2}{2}\right) \geq 1.
\ee
Notably, the above relation coincides with the Maassen-Uffink uncertainty relation \cite{MaasenUfff1988} of $\sigma_x,\sigma_z$.
By taking values of the parameters $r_{1,2},s_{1,2}$ in small interval, one can see that the above relation \eqref{eq:icp_to_ur} is satisfied if 
\be \label{r1r2s1s2}
 r_1+r_2+s_1+s_2 \geq 0.44.
\ee 
This relation is valid for any two given preparations $P_{11},P_{12}$. Thanks to the symmetry, there exists a preparation on the diagonal of the square that corresponds to the minimum uncertainty of all possible preparations, i.e., the uncertainty of $\setth$. Again, exploiting the symmetry one knows that the origin of the largest square fit inside $\setth$ is the center of the square. Therefore, for the symmetric statistics set,
\be \label{k:symmetric}
\U = 2\min(r,s), \ \C_r = 1 - 2\max(r,s)
\ee where $r_1=r_2=r, s_1=s_2=s$.
Subsequently, it follows from \eqref{r1r2s1s2} that $\C_r - \U \leq 0.56$ which captures the PUR. Namely, the last formula 
says that for strong enough complementarity uncertainty must appear.

\subsection{Tsirelson bound from uncertainty principle and non-signaling}
\label{sec:nonlocality}
Here, we discuss how Uncertainty principle in a theory sets restriction on nonlocality of that theory.  We concentrate on the simplest scenario of nonlocality where two spatially separated parties, Alice and Bob, perform one of the two binary outcome measurements $A_{1,2},B_{1,2} \in \{+,-\}$ on their respective subsystems of a bipartite system. The witness based on the measurement statistics of nonlocality is taken to be the violation of well-known Clauser-Horne-Shimony-Holt (CHSH) local-realist inequality \cite{CHSH},
\be \label{chsh}
\mathcal{I} = \<A_1B_1\> + \<A_1B_2\> + \<A_2B_1\> - \<A_2B_2\> \leq 2.
\ee
Without loss of generality, we can say that, Bob's measurement statistics of the observables $B_1,B_2$ on his system are $\q_{B_1}(P),\q_{B_2}(P)$ for some $P$ when Alice does not perform any measurement. As a result of sharing correlated systems, depending on Alice's measurement choice and outcome the preparation on Bob's side might be different. In other words, Alice's measurement steers different preparation on Bob's subsystem. Let us denote Bob's preparation as $P_{A_1+}$ if Alice measures $A_1$ and obtains $+$ outcome on her subsystem and so on. This phenomenon is called as `steering' \cite{OW}. However, we do not impose any restriction on steering, except the no-signaling principle which should be satisfied by any physical theory. The 'no-signaling' principle is a direct consequence of relativistic causation, which says that, Alice cannot send any information to Bob instantaneously. That is, the measurement statistics on Bob's subsystem is independent on the Alice's measurement choice and vice-versa. Formally, $\forall i\in \{1,2\},$
\beq \label{ns}
&& \q_{B_{i}}(P) = q_{A_1}(+|\tilde{P}) \q_{B_{i}} (P_{A_1+}) +  q_{A_1}(-|\tilde{P}) \q_{B_{i}} (P_{A_1-}) \nonumber \\
&& = q_{A_2}(+|\tilde{P}) \q_{B_{i}} (P_{A_2+}) +  q_{A_2}(-|\tilde{P}) \q_{B_{i}} (P_{A_2-})
\eeq
where $\tilde{P}$ denotes Alice's initial preparation. For simplicity, we denote,
\beq
&q_{A_1}(+|\tilde{P}) = t_1, \ q_{A_2}(+|\tilde{P}) = t_2, \nonumber  \\
&q_{B_{1}} (+|P_{A_1+}) = 1-r_1, \ q_{B_{2}} (+|P_{A_1+}) = 1-r_2, \nonumber  \\
&q_{B_{1}} (+|P_{A_1-}) = r'_1, \ q_{B_{2}} (+|P_{A_1-}) = r'_2, \nonumber \\
&q_{B_{1}} (+|P_{A_2+}) = 1-s_1, \ q_{B_{1}} (+|P_{A_2+}) = s_2,  \nonumber  \\
&q_{B_{2}} (+|P_{A_2-}) = s'_1, \ q_{B_{2}} (+|P_{A_2-}) = 1-s'_2, 
\eeq
as shown in Fig. \ref{figchsh}. Subsequently, the CHSH term is expressed as follow,
\beq \label{chshsim}
&\mathcal{I}=& q_{A_1}(+|\tilde{P}) (2- 2q_{B_{1}} (-|P_{A_1+}) -2q_{B_{2}} (-|P_{A_1+})) \nonumber \\
&& + q_{A_1}(-|\tilde{P}) (2- 2q_{B_{1}} (-|P_{A_1-}) -2q_{B_{2}} (-|P_{A_1-})) \nonumber  \\
&& + q_{A_2}(+|\tilde{P}) (2- 2q_{B_{1}} (-|P_{A_2+}) -2q_{B_{2}} (+|P_{A_2+})) \nonumber \\
&& + q_{A_2}(-|\tilde{P}) (2- 2q_{B_{1}} (-|P_{A_2-}) -2q_{B_{2}} (+|P_{A_2-})) \nonumber \\
&= & 4 -2\big(t_1(r_1+r_2)+(1-t_1)(r'_1+r'_2) + t_2(s_1+s_2) \nonumber \\ 
&& +(1-t_2)(s'_1+s'_2) \big).
\eeq
While the no-signaling conditions simplify to,
\beq \label{nssim}
& t_1(1-r_1)+(1-t_1)r'_1 = t_2(1-s_1) + (1-t_2) s'_1, \nonumber \\
& t_1(1-r_2)+(1-t_1)r'_2 = t_2 s_2 + (1-t_2) (1-s'_2).
\eeq 

\begin{figure}[http]
\begin{center}
\includegraphics[scale=0.45]{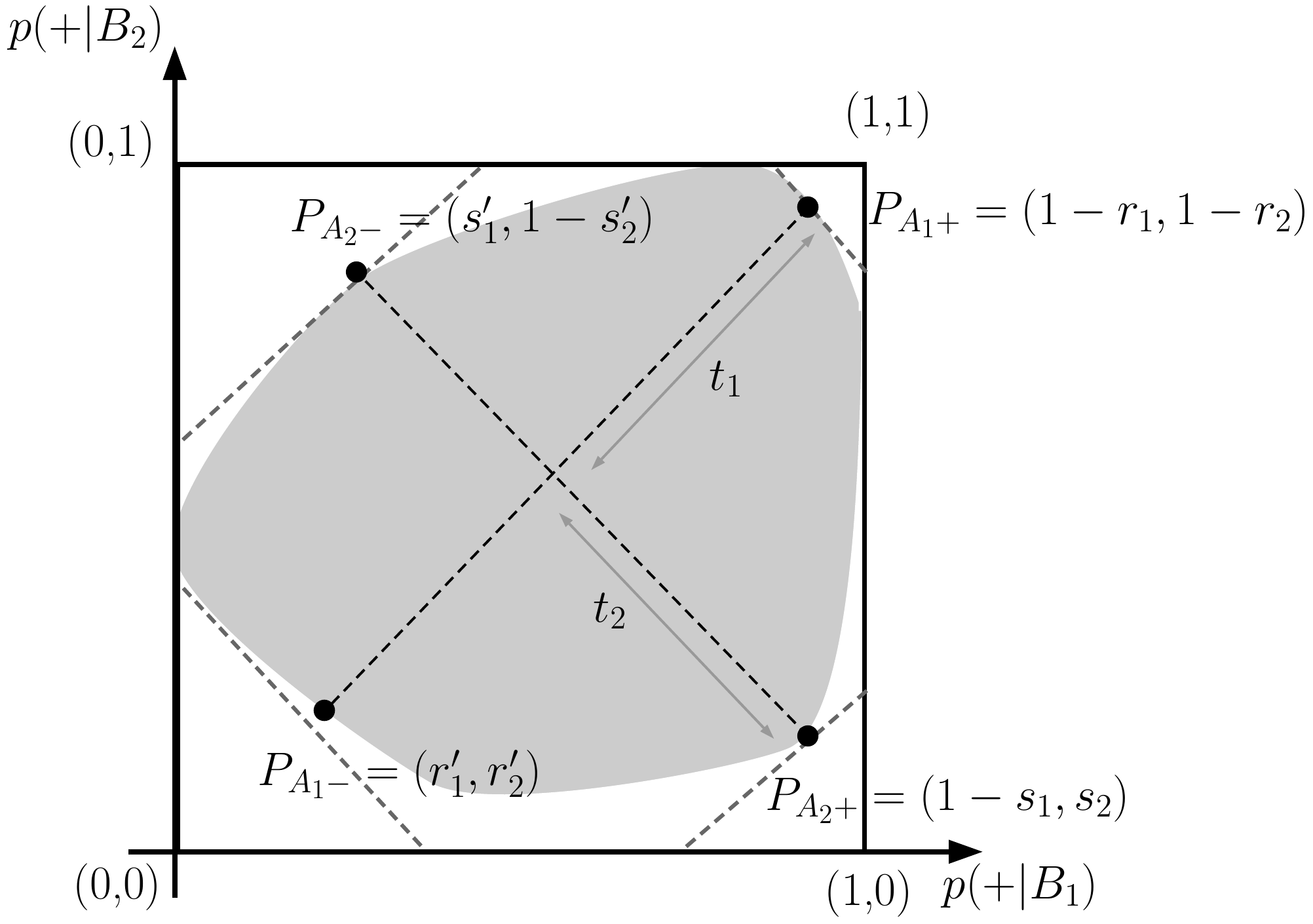}
\end{center}
\caption{An arbitrary statistics set $\setth$ for two observables $B_{1,2}$ of Bob's system. The four different preparations $P_{A_i\pm}$ depending on Alice's measurement choice and outcome are presented by their coordinates. The four preparations should satisfy the no-signaling conditions \eqref{nssim}.}
\label{figchsh}
\end{figure}

 Thus, we seek to maximize the right-hand-side of \eqref{chshsim} under the non-linear constraints \eqref{nssim}. Intuitively, it can be seen that the PUR prevents the CHSH value to be the maximum. There are only few possibilities for $\mathcal{I} =4$. In one case,  the statistics set allows the four corners of the square, i.e., $r_1+r_2=r'_1+r'_2=s_1+s_2=s'_1+s'_2=0$, which contradicts the notion of PUR. On the other, one of terms $r_1+r_2$ or $r'_1+r'_2$, say $r_1+r_2$, and one of terms $s_1+s_2$ or $s'_1+s'_2$, say $s_1+s_2$, is zero and accordingly $t_1,t_2$ both has to be 1. Such value assignment of these variables contradicts with no-signaling principle \eqref{nssim}.

If we assume $S_{X,Y}$ to be symmetric with respect to the diagonal of the square (as shown in Fig. \ref{fig:icp}), then it is easier to relate the CHSH term \eqref{chshsim} to Uncertainty principle. Consider $P_{A_1+},P_{A_2+}$ to be the closest points to the corners $(1,1)$ and $(1,0)$ respectively. By symmetry, we know there exists another two closest points to other two corners, such that $r_1=r'_1, r_2=r'_2,s_1=s'_1,s_2=s'_2$. Therefore, $\mathcal{I} \leq 4-2(r_1+r_2+s_1+s_2)$. In fact, this inequality is tight, due to the fact that, this value is achieved when the no-signaling conditions \eqref{nssim} are satisfied for $t_1=t_2=1/2$. Further, we recall the expression of $\U,\C_r$ from \eqref{k:symmetric} in terms of $r_1,r_2,s_1,s_2$, and re-express the CHSH term as,
\be \label{ur-chsh}
\mathcal{I} = 2 + 2(\C_r - \U).
\ee
Clearly, the Uncertainty principle, which is in the form \eqref{pur}, restricts the value of $\mathcal{I}$. In quantum theory, the exact form of PUR is given in \eqref{eq:urq2}. Thus, the maximum value of RHS of \eqref{ur-chsh} is obtained to the Tsirelson's bound, i.e., $2\sqrt{2}$, when $\C_r=1-\U = 1/\sqrt{2}$ satisfying \eqref{eq:urq2}.


\section{Open problems}\label{sec:openPROB}
The major open problem is whether there exists theories, where two clean \mg and extremal \blk observables can be very well approximated by 
some other observable. For such hypothetical theories, complementarity of observables cannot be anymore read out from behaviour 
of the statistics set. It wold be also interesting to define a smoothed version of complementarity,
given by minimum of independence over observables that reproducing the given observables up to $\epsilon$ in some suitable 
distance.  One can then investigate how the statistics set changes with $\epsilon$. 
Another interesting problem is to explore the relation between the concepts of complementarity and contextuality \cite{SpekkensCont2005}, as the latter also reflects somehow the notion of complementarity. 
There is also a question of how the approach presented in this paper are related to the operational approach to wave particle duality 
of Ref. \cite{Bagan2018-duality}.

There are lot of other questions, including the following ones:
\begin{itemize}
\item Generalize the geometric approach to continuous variables (i.e. to position and momentum observables). 
\item Prove that uncertainty relation implies Tsirelson bound without symmetry assumptions. 
\item Relate Information Contents Principle to uncertainty relation for larger dimensions, and again, without symmetry assumptions. 
\item Compute independence based on variation of information for qubit observables, and find uncertainty relation with properly chosen uncertainty measure (seems that in this case entropy is the suitable one, or mutual information as exclusion measure in higher dimension) 
\item Make tighter exclusion principle based on random access code. 
\end{itemize}
Finally, our focused exclusively on two observables, but one can readily extend the definitions and concepts to  more observables 
and explore the subject in this more general setting.

\emph{Note added---} During the completion of our manuscript, we became aware of the paper \cite{ComplNonloc2018}, that  derived the Tsirelson bound for CHSH inequality from restrictions on the complementarity present in quantum theory.    However, the quantitative notion of complementarity used in that work differs form considered by us.

\begin{acknowledgments}
We thank Karol Horodecki interesting and stimulating discussions. 
D. S. is supported by National Science Centre, Poland, grants 2016/23/N/ST2/02817,
2014/14/E/ST2/00020 and FNP grant First TEAM (Grant No. First TEAM/2017-4/31). 
M.O. acknowledges the support of Homing programme of the Foundation for Polish Science co-financed
by the European Union under the European Regional Development Fund. 
L.Cz., M.H. and R.H. are supported by John  Templeton Foundation through grant  ID \#56033. M.H. and R.H. are also supported 
by National Science Centre, Poland, grant OPUS 9. 2015/17/B/ST2/01945.
\end{acknowledgments}

\bibliography{CompUncBIB}

\appendix

\part*{Appendices}

In the appendices we present proofs of technical results that were omitted the main text.

\section{Proof of qualitative uncertainty relations}
\label{sec:appendix-implications}
\begin{lem}
\label{lem:appendix-implications}
In quantum mechanics, for quantum measurements with one dimensional projectors the following statements hold
\begin{itemize}
\item[(i)] Complementarity implies information exclusion
\item[(ii)] Single-outcome complementarity implies  uncertainty
\item[(iii)] Full complementarity implies uncertainty
\end{itemize}
\end{lem}

\begin{proof}
We prove each implication individually.
\begin{itemize}
\item[Ad. (i)] Note first that two rank one projective measurements $X$ and $Y$ are not jointly measurable if and only if  they do not commute. In other words  some projector $P_i$ of $X$  and some projector $Q_j$ of $Y$ do not commute (see e.g. 
\cite{Heinosaari2008} for the proof of this statement).  Now suppose, by contraposition, that there is no exclusion for $X$ and $Y$.  This means the $d$ states that have distinct deterministic  outcomes for observable $X$ and $Y$. Hence,  the states are distinct eigenstates of the both observables. Therefore, $X$ and $Y$ commute, hence they are not complementary. 

\item[Ad. (ii)] Again by contraposition, suppose that there is no uncertainty. This means that the observables share a common eigenvector. 
Consider coarse graining for both observables: this vector versus the complement. Clearly the new binary observables are the same,
hence  do not exhibit complementarity.  Hence, by definition, the original observables do not exhibit single-outcome complementarity. 

\item[Ad. (iii)] Full complementarity by definition is a stronger notion than single-outcome complementarity. Therefore, (ii) implies (iii). 
\end{itemize}
\end{proof}

\section{Full complementarity does not imply strong uncertainty}\label{app:fullCOMP}
We will now give the example of two fine-grained projective measurements in $\mathbb{C}^5$ that do not exhibit full preparation uncertainty even though they are fully complementary. We consider two orthonormal bases (for brevity we write unnormalized vectors)
\begin{eqnarray}
&&\ket{\psi_1}=|0\>\ ,\ \ket{\psi_2}=|1\>\ ,\ \ket{\psi_3}=|2\>\ , \nonumber \\
&&\ket{\psi_4}=|3\>+|4\>\ ,\ \ket{\psi_5}=|3\>-|4\>\ , 
\end{eqnarray}
and 
\begin{eqnarray}
&&\ket{\phi_1}=|0\>+|1\>\ ,\ \ket{\phi_2}=|0\>-|1\>+|2\>\ ,\ \ket{\phi_3}=|3\>+\ket{\chi}\ , \nonumber \\
&&\ket{\phi_4}=|4\>\ ,\ \ket{\phi_5}=|3\>-\ket{\chi}\ , 
\end{eqnarray}
where $\ket{\chi}=(|0\>-|1\>-2|2\>)/\sqrt{6}$. One readily checks that the following coarse grainings:
\begin{eqnarray}
&&P_1 = \sum_{i=1}^3 \ketbra{\psi_1}{\psi_1} \ ,\  P_2 = \sum_{i=4}^5 \ketbra{\psi_1}{\psi_1}\ , \nonumber \\
&&Q_1 = \sum_{i=1}^3 \ketbra{\phi_1}{\phi_1} \ ,\  Q_2 = \sum_{i=4}^5 \ketbra{\phi_1}{\phi_1}\ ,
\end{eqnarray}
do not exhibit uncertainty, as the input states $\ket{\psi}=(1/\sqrt{5})(2|0\> + |2\>)$ gives deterministic outcome for both (now binary) measurements. Specifically, this state gives with certainty the outcomes corresponding to projector $P_1$ and $Q_1$ respectively. On the other hand, arbitrary coarse graining of the fine grained measurements lead to non-commuting projectors and therefore by \cite{Heinosaari2008}  are jointly non-measurable projective measurements. Hence the above two measurements, although do not exhibit strong uncertainty, are fully complementary.

\section{For non-extremal observables \precom\ does not capture \com}\label{appendix:non-extremal}

In this section, we argue that \precom\ is not a good indicator of \com\ for non-extremal observables. Particularly, we provide an example where the \precom\ increases under taking convex mixture of observables. Consider a theory containing  three 3-outcome  observables $X_1,X_2,Y$ whose statistics sets origine from convex combinations of three preparations $P_1,P_2,P_3$ such that,
\beq \label{XiYstat}
&&\q_{X_1}(P_1) = (1,0,0),\q_{X_1}(P_2) = (0,1,0), \q_{X_1}(P_3) = (0,0,1), \nonumber \\
&&\q_{X_2}(P_1) = (\frac14,0,\frac34),\q_{X_2}(P_2) = (\frac34,0,\frac14), \q_{X_2}(P_3) = (0,1,0), \nonumber \\
&&\q_{Y}(P_1) = (\frac14,\frac34,0), \q_{Y}(P_2) = (\frac34,\frac14,0), \q_{Y}(P_3) = (0,0,1). \nonumber \\
\eeq 
We can verify there exists two left-stochastic maps,
\beq
\Lambda_1 =  
\left[ {\begin{array}{ccc}
   \frac14 & \frac34 & 0 \\
   \frac34 & \frac14 & 0 \\
   0 & 0 & 1
  \end{array} } \right], \ 
  \Lambda_2 =  
\left[ {\begin{array}{ccc}
   1 & 0 & 0 \\
   0 & 0 & 1 \\
   0 & 1 & 0
  \end{array} } \right]
\eeq
for which $\q_Y(P_i) = \Lambda_1 \q_{X_1}(P_i) = \Lambda_2 \q_{X_2}(P_i) $, thereby $X_{1,2} \rightarrow Y$. Consider another observable $X$ as a convex mixture of $X_1$ and $X_2$ with equal probability. From \eqref{XiYstat} we obtain, 
\be \label{Xstat}
\q_{X}(P_1) = (\frac58,0,\frac38), \q_{X}(P_2) = (\frac38,\frac12,\frac18), \q_{X}(P_3) = (0,\frac12,\frac12). 
\ee
Let us assume there exists a left-stochastic map,
\be
\Lambda = 
\left[ {\begin{array}{ccc}
   t_{11} & t_{12} & t_{13} \\
    t_{21}& t_{22} & t_{23} \\
   t_{31} & t_{32} & t_{33}
  \end{array} } \right] \nonumber
 \ee
such that $\Lambda \q_{X}(P_i) = \q_{Y}(P_i)$. From \eqref{XiYstat}-\eqref{Xstat} we see $\Lambda \q_{X}(P_3) = \q_{Y}(P_3)$ implies $t_{22}=t_{23}=0$. Further, imposing this condition on $\Lambda \q_{X}(P_1) = \q_{Y}(P_1)$, we obtain $\frac58 t_{21}=\frac34$ that implies $t_{21}=\frac65>1$. This is not possible for a left-stochastic map $\Lambda$. Similarly, if we assume $\Lambda \q_{Y}(P_i) = \q_{X}(P_i)$, we can check that $\Lambda \q_{Y}(P_1) = \q_{X}(P_1)$ implies $t_{21}=t_{22}=0$, however $\Lambda \q_{Y}(P_2) = \q_{X}(P_2)$ suggests $\frac34 t_{21}+\frac14 t_{22}=\frac12$. Hence, such a stochastic map does not exist. In other words, \precom\ of $X,Y$ is non-zero.

\section{Proof of the optimal success probability in random access code for two projective measurements }\label{appendix:qrac}

We consider two quantum projective measurements correspond to the basis $X_1 = \{|i\rangle\}^{d}_{i=1}$ and $X_2 = \{|\psi\rangle_j\}^{d}_{j=1}$ accessed by Bob. Given Alice's input $a_1a_2$ and her encoding state $\rho_{a_1a_2}$, the success probability of guessing $a_y$ is,
\be \label{eq:qrac1}
\sum_b p(a_b|a,b) = tr((|a_1\>\<a_1| + |\psi_{a_2}\>\<\psi_{a_2}|)\rho_{a_1a_2}).
\ee
Since, the operator $|a_1\>\<a_1| + |\psi_{a_2}\>\<\psi_{a_2}|$ is hermitian, its eigen vectors span $d$-dimensional space. The optimal value of the RHS \eqref{eq:qrac1} is maximum eigenvalue of this operator and $\rho_{a_1a_2}$ is the corresponding eigenvector. A simple calculation leads to the fact that the maximum eigenvalue of  $|a_1\>\<a_1| + |\psi_{a_2}\>\<\psi_{a_2}|$ is $1+|\<a_1|\psi_{a_2}\>|$. Subsequently, the average success probability \eqref{avgrac} is,
\be \label{eq:qps}
p_s(X_1,X_2) = \frac{1}{2} + \frac{1}{2d^2} \sum_{a_1,a_2} |\<a_1|\psi_{a_2}\>|.
\ee 
To show that the above expression is the optimal success probability given the two measurements $X_1,X_2$, we need to show that any classical post-processing of the outcome statistics will not yield higher success probability. Any post-processing can be represented by the set of positive operators $\{ M_{a_1} \}^d_{a_1=1}$ and $\{ M_{a_2} \}^d_{a_2=1}$, corresponds to $y=1,2$ respectively, as follows
\be
M_{a_1} = \sum^d_{i=1} p(a_1|i) |i\>\<i|, \  M_{a_2} = \sum^d_{j=1} q(a_2|j) |\psi_{j}\>\<\psi_{j}| ,
\ee 
for some probability distributions such that $\forall i,j,\ \sum_{a_1} p(a_1|i) = \sum_{a_2} q(a_2|j)  =1$. 
Since $ M_{a_1} + M_{a_2}$ is a positive operator, following the previous argument we know the optimal success probability for this strategy is,
\be
p_s = \frac{1}{2d^2} \sum_{a_1,a_2} ( || M_{a_1} + M_{a_2} || )
\ee 
where $||M||$ denotes the operator norm. Using the inequality $|| X+Y|| \leq \max(||X||, ||Y||) + ||\sqrt{X}\sqrt{Y}||$ derived by Kittaneh \cite{Kittaneh} and the fact $|| X+Y|| \leq ||X|| + ||Y||$,  we obtain the following relation,
\beq
&p_s& = \frac{1}{2d^2} \sum_{a_1,a_2} ( || M_{a_1} + M_{a_2} || ) \nonumber \\
&& \leq \frac{1}{2d^2} \sum_{a_1,a_2} \left( \max(||M_{a_1}||, ||M_{a_2}||) + || \sqrt{M_{a_1}} \sqrt{M_{a_2}} || \right) \nonumber \\
&& \leq  \frac{1}{2d^2} \sum_{a_1,a_2} ( 1 + ||  \sum_{i,j} \sqrt{p(a_1|i)}\sqrt{q(a_2|j)}\ |i\>\<i|\psi_j\>\<\psi_j|\ ||) \nonumber \\
&& \leq  \frac{1}{2} + \frac{1}{2d^2} \sum_{a_1,a_2} \sum_{i,j} \sqrt{p(a_1|i)}\sqrt{q(a_2|j)} || |i\>\<i|\psi_j\>\<\psi_j| \ || \nonumber \\
&& \leq  \frac{1}{2} + \frac{1}{2d^2} \sum_{a_1,a_2}  \sum_{i,j} p(a_1|i)q(a_2|j) \ |\<i|\psi_j\>| \nonumber \\
&& = \frac{1}{2} + \frac{1}{2d^2} \sum_{i,j} |\<i|\psi_j\>| 
\eeq
 which is same as the left-hand-side of \eqref{eq:qps}. In the above derivation, we have used the fact that $\sqrt{M_{a_1}} = \sum^d_{i=1} \sqrt{p(a_1|i)} |i\>\<i|, \sqrt{  M_{a_2}} = \sum^d_{j=1} \sqrt{q(a_2|j)} |\psi_{j}\>\<\psi_{j}| $.

\section{Proof of exclusion relation from random access code}

\begin{lem}[Quantum-mechanical uncertainty relation for exclusion-like quantity defined in terms of RAC]
Consider a $d$ dimensional quantum system and let $X=\lbrace \ket{i} \rbrace_{i=1}^d$ and $Y= \lbrace \ket{\psi_i} \rbrace_{i=1}^d$ be two projective measurements in $\mathbb{C}^d$.  Let $E(S_{X,Y})$ be the exclusion-like quantum-mechanical quantity given in \eqref{eq:E-rac} and let $Ind(S_{X,Y})$ be the quantum-mechanical complementarity measure based on the average success probability in RAC given in \eqref{eq:ind-rac} Then, the following uncertainty relation holds
\begin{equation}\label{eq:exclusionURrac}
E(S_{X,Y})\geq \frac{Ind^2(S_{X,Y})}{4d}
\end{equation}
\end{lem}

\begin{proof}
In what follows we will use the notation $U_{ij}= \braket{i}{\psi_j}$.  Recall the explicit formulas for $E(S_{X,Y})$ and $(S_{X,Y}$, 
\begin{eqnarray}
E(S_{X,Y})& = \frac{1}{2}\left(1- \frac{1}{d} \max_\pi \sum_{i=1}^d |U_{i\pi(i)}|  \right) \ , \label{eq:excQapp} \\
Ind(S_{X,Y})& =\frac{1}{d-1}\left(\frac{1}{d}\sum_{i,j=1}^d |U_{ij}| -1\right)\ .  \label{eq:complQapp}
\end{eqnarray}
Let $p_\mathrm{max}(i) \eqdef \max_j |U_{ij}|^2$.  where maximum is over $j\in\lbrace 1,\ldots,d\rbrace$. Our proof strategy is to show that $I(S_{X,Y})>0$ implies $p_\mathrm{max}(i_0)<1$ for some $i_0$. As we will prove later the latter condition can be used to find lower bound on the exclusivity $E(S_{X,Y})$. By reformulating Eq.\eqref{eq:complQapp} we obtain 
\be
\sum_{i,j=1}^d  |U_{ij}|  = (d-1)d I(S_{X,Y})+d\ ,
\ee
from which we can readily deduce that for some $i_0$ we have the inequality  $\sum_{j=1}^d  |U_{i_0j}| \geq 1+(d-1)Ind(S_{X,Y})$. The LHS of this inequality can be upper bounded as
\begin{equation}\label{eq:intINEQ}
\sum_{j=1}^d  |U_{i_0j}| \leq \sqrt{p_\mathrm{max}(i_0)}+\sqrt{d-1}\sqrt{1-p_\mathrm{max}(i_0)}\ ,  
\end{equation}
where we have used the Shur-concavity of the square-root function and the fact that for fixed $i$ numbers $|U_{ij}|^2$ form a probability distribution. Combining \eqref{eq:intINEQ} with the earlier bound gives 
\begin{equation}
\sqrt{p_\mathrm{max}(i_0)}+\sqrt{d-1}\sqrt{1-p_\mathrm{max}(i_0)} \geq 1 +(d-1) Ind(S_{X,Y})\ .
\end{equation}
Importantly, the above inequality implies that $p_\mathrm{max}(i_0)<1$ whenever $Ind(S_{X,Y})>0$. To get a nontrivial upper bound on $p_\mathrm{max}(i_0)$ we apply the inequality\footnote{This follows form the simple inequality $\sqrt{1-y}\leq 1-(y/2)$ used for $y=1-p_{\mathrm{max}}(i_0)$}  $\sqrt{p_{\max}(i_0)} \leq 1 -(1/2)(1-p_\mathrm{max}(i_0))$ which finally gives
\begin{equation}
\sqrt{d-1}y-(1/2)y^2 \geq (d-1) Ind(S_{X,Y})\ ,
\end{equation}   
for $y=\sqrt{1-p_\mathrm{max}(i_0)}$. By neglecting the quadratic we obtain
\begin{equation}\label{eq:finalBOUND}
p_{\max} (i_0) \leq 1 -(d-1)Ind^2(S_{X,Y})\ .
\end{equation}
 
To conclude we prove a lower bound on $E$ in terms of $p_{\max}(i_0)$ we note that the following inequalities hold true
\begin{equation}
\max_\pi \sum_i|U_{i\pi(i)}| \leq \sum_{i=1}^d \sqrt{p_\mathrm{max}(i)}\leq (d-1) + \sqrt{p_\mathrm{max}(i_0)}\ .
\end{equation}
Plugging this bound into \eqref{eq:excQapp} gives $E(S_{X,Y})\geq 1- \sqrt{p_{\max}(i_0)}/(2d)$. Together with  \eqref{eq:finalBOUND} this gives 
\begin{equation}
E(S_{X,Y})\geq \frac{1}{2d}\left(1-\sqrt{1-(d-1) Ind(S_{X,Y}^2)}\right) \ .
\end{equation}
Using inequality $1 - \sqrt{1-x}\geq x/2$ valid for all $x\in(0,1)$ we obtain the final result
\be
E(S_{X,Y})\geq \frac{Ind^2(S_{X,Y})}{4d}\ .
\ee Of course, since  the considered observables are clean and extremal $Ind$ is same as complementarity. 
\end{proof}

\section{Proof of reverse uncertainty relation} \label{app:reverseUR}

\textit{For $S$ being square for two clean and sharp observables, any theory must satisfy the reverse PUR given by $2\C_r \geq \U$. }\\
\begin{proof}
 Since the observables are sharp, the statistics set $\setth$ touches all the four edges of the square $S$. Let us say the minimum distance between the corners and the points belong to $\setth$, that lie on the boundary of $S$, is $t$ (see Fig. \ref{fig:rur}). It is clear from the definition that the re-scaling measure of uncertainty of that point is $t$, and hence the uncertainty measure of the statistics set $\U \leq t$. Now, consider a square of length $t$ taking the same origin of $S$. As described in the Fig. \ref{fig:rur}, this square should always fits inside $\setth$. This leads to the fact that $\C_r \geq t/2$, and subsequently $2\C_r\geq \U$. 
\end{proof}
\begin{figure}[http]
\begin{center}
\includegraphics[scale=0.71]{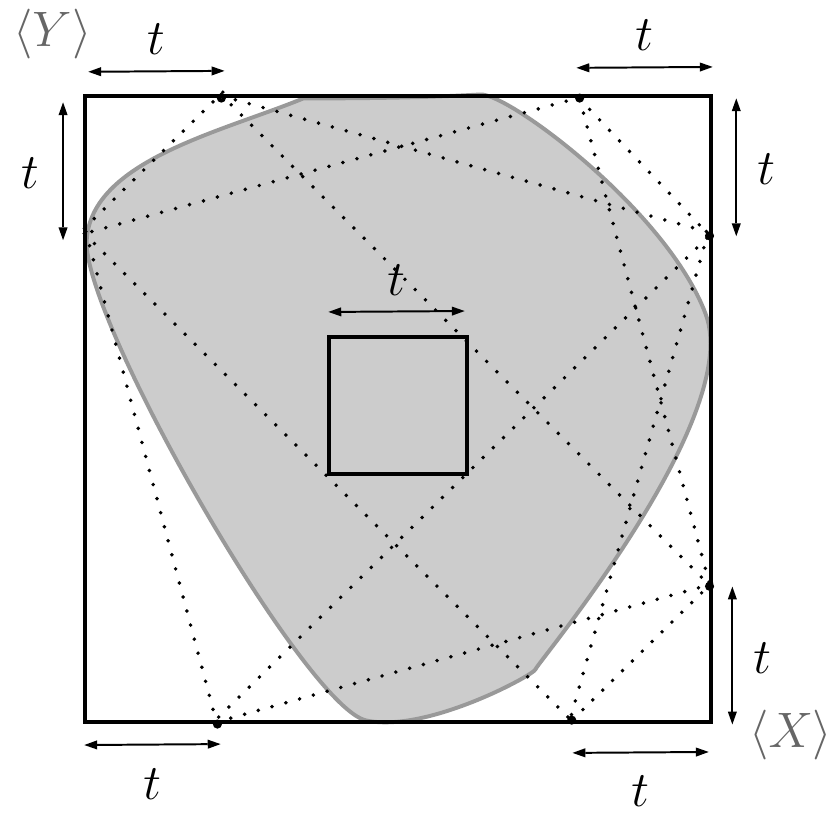}
\end{center}
\caption{We assume that the minimum distance between the corners and the points on $\setth$, that lie on the boundary of $S$, is $t$. This, together with the fact that $\setth$ touches all the four edges of square, impose the boundaries of $\setth$ cannot be closer to the center than the dotted lines presented here. This implies that a square of length $t$ will fit inside $\setth$.}
\label{fig:rur}
\end{figure}

\end{document}